\documentclass[10pt]{article}

\usepackage[utf8]{inputenc}
\usepackage[english]{babel}
\usepackage[T1]{fontenc}
\usepackage{amsfonts}
\usepackage{amsmath}
\usepackage{amssymb}
\usepackage{amsthm}
\usepackage{graphicx}
\usepackage{pstricks}
\usepackage{dsfont}
\usepackage{xcolor}
\usepackage{bm}
\usepackage{enumitem}
\usepackage{ulem}
\usepackage{verbatim}

\usepackage{thmtools}
\usepackage[pdfdisplaydoctitle,colorlinks,breaklinks,urlcolor=blue,linkcolor=blue,citecolor=blue]{hyperref}

\newtheorem{lemma}{Lemma}[section]
\newtheorem{theorem}[lemma]{Theorem}
\newtheorem{proposition}[lemma]{Proposition}
\newtheorem{corollary}[lemma]{Corollary}

\newtheorem{remark}[lemma]{Remark}
\newtheorem{hypothesis}[lemma]{Assumption}
\newtheorem{example}[lemma]{Example}

\newtheorem{notation}{Notation}

\newcommand{\R}{\mathbb{R}}

\def\<{\left\langle }
\def\>{\right\rangle }

\newhsbcolor{mygreen}{.3 1 .6}

\def\ds{\begin{displaystyle}}
\def\eds{\end{displaystyle}}

\def\R{\mathbb R}

\def\E{\mathbb E}
\def\P{\mathbb P}

\def\F{\mathbb F}

\def\1{\mathbf 1}
\def\to{\rightarrow}

\begin{document}

\title{\bf A Pricing Formula for Delayed Claims: \\Appreciating the Past to Value the Future}
%

\author{Enrico Biffis\footnote{Enrico Biffis (\texttt{e.biffis@imperial.ac.uk}), Imperial College Business School, Imperial College London, South Kensington Campus, London SW7 2AZ, UK.}
\and Beniamin Goldys\footnote{Beniamin Goldys (\texttt{beniamin.goldys@sydney.edu.au}), School of Mathematics and Statistics, The University of Sydney, NSW 2006, Australia.}
\and Cecilia Prosdocimi\footnote{Cecilia Prosdocimi (\texttt{cecilia.prosdocimi@gmail.com}), LUISS "Guido Carli", Rome, Italy.}
\and Margherita Zanella\footnote{Margherita Zanella (\texttt{margherita.zanella@polimi.it}), Politecnico di Milano, Dipartimento di Matematica `Francesco Brioschi', Via Bonardi 13, 20133 Milano, Italy.}}

\date{\today}

\maketitle

\begin{abstract}
\noindent We consider the valuation of contingent claims with delayed dynamics in a
Black\&Scholes complete market model. We find a pricing formula that
can be decomposed into terms reflecting the market values
of the past and the present, showing how the valuation of future cashflows
cannot abstract away from the contribution of the past.
As a practical application, we provide an explicit expression for the market value of human capital in a setting with wage rigidity. The formula we derive has successfully been used to explicitly solve the infinite dimensional stochastic control problems addressed in \cite{BGP}, \cite{BiaGozZan} and \cite{DGZZ} .
\end{abstract}


\bigskip

\bigskip


\bigskip

\bigskip

\noindent


\tableofcontents

\section{Introduction}
It is a standard result in asset pricing theory that the absence of arbitrage opportunities
is essentially equivalent to the existence of an equivalent probability measure 
under which
the price of any contingent claim is a local martingale after deflation by
the money market account; see \cite{DS1994,HK1979,HP1981}. In this paper we preserve the standard formulation of arbitrage pricing in a complete market model with security prices evolving as geometric Brownian motions (GBM). The main novelty of our work is that we consider contingent claims that have dynamics described by a stochastic delay differential equation (SDDE). 
\par
 It is perhaps surprising that using the no-arbitrage
pricing machinery we are able to derive an explicit valuation formula in the case of dynamics with memory, which is notoriously difficult to study. In particular, we show that the price can be decomposed into a term related
to the `current market value of the past' (in a sense to be made precise below),
and a term reflecting the `market value of the present'. 
In our setting the contribution of the past is represented
by the portion of a contingent claim's past trajectory that shapes its dynamics
going forward.\footnote{The importance of the past in understanding
the qualitative feature of a model with delay was also emphasized in Fabbri and Gozzi \cite{FG2008},
although in a deterministic setting, when solving the endogenous growth
model with vintage capital of Boucekkine {\it et al.} \cite{boucekkine2005vintage}. }
Using our pricing formula, we demonstrate that 
in the market consistent valuation of future cashflows the contribution of the past cannot be neglected.
\par
As a practical application of our results, we consider in detail the case in which the 
contingent claim represents stochastic wages received by an agent
over her lifetime (e.g., \cite{DL2010,BGP}). 
 It is well known that when labor income
is spanned by tradable assets, the market value of human capital can be easily
derived via risk-neutral valuation. 
The empirical literature on labor income dynamics widely relies on auto-regressive moving average (ARMA) processes  (e.g., \cite{MaCurdy1982}, \cite{AbowdCard1989}, \cite{Hubbard1995}, \cite{meghir2004income}); however
\cite{Reiss2002}, \cite{Lorenz2005}, and \cite{DGT2015} show how SDDEs can be understood as the weak limit of discrete time ARMA processes.
Therefore specific classes of  SDDEs can be chosen as  models for labor income, since they are continuous time counterparts of the empirically validated ARMA models.	
We thus consider the introduction of delayed drift and volatility coefficients
in a GBM labor income model to provide a tractable example of 
wage dynamics that adjusts slowly to financial market shocks.
We obtain a closed form solution for human capital, which makes
explicit the contributions of the market value of the past and the present. Our results
demonstrate that SDDEs are valuable modeling tools that can address the findings 
of a large body of empirical literature on wage rigidity (e.g., 
\cite{khan1997}, \cite{dickens2007}, \cite{barattieri2010}, \cite{lebihan2012}). 
Moreover, our results become an essential ingredient in finding explicit solutions of an interesting class of agent's lifecycle portfolio choice problems with no-borrowing-without repayment constraint, where the agent receive labor income whose dynamics is path dependent (see Section \ref{Applications}).

Although we discuss the human capital application extensively, the  
extension to other applications is immediate. For instance, we provide some references 
to the literature on counterparty risk and derivatives valuation, in which analogous dynamics 
arise in the context of collateralization procedures entailing a delay in the marking-to-market
procedure of over-the-counter derivatives (e.g., \cite{brigo2013counterparty,BrigoPallavicini2014}).
\par
It should be noted that no-arbitrage pricing in the case of delayed price dynamics 
has been recently studied by many authors, see for example \cite{arriojas2007delayed,mao2013delay}. Their focus however is on proving completeness of the market, hence very different from ours. On the other hand, their work suggests that our results are of broader applicability, in particular to settings where market completeness is preserved, such as the case in which tradable assets have delayed drift and volatility terms (we refer to Section \ref{sec:conclusions} for more comments). 
\par

The paper is organized as follows. In Section \ref{Section_Mathematical_tools} we introduce the
setup and state our main result. In Section~\ref{proof_section} we prove   the result. 
 In Section~\ref{Applications} we provide some applications.
Section~\ref{sec:conclusions} concludes. In the Appendix we present the proof of some more technical results.

\begin{notation}
Given $a, b \in \mathbb{R}^n$, by $a \cdot b$ we denote the scalar product in $\mathbb{R}^n$ and by $\| \cdot\|$ we denote the corresponding norm. Given a signed measure of bounded variation $\phi$ on $[-d,0]$, by $|\phi|$ we denote its total variation.
\newline
The space of Lebesgue square integrable (deterministic) functions on $[-d,0]$ is denoted by $L^2(-d,0;\mathbb{R})$, that is the measure ${\rm d}t$ is understood. By $W^{1,2}(-d,0;\mathbb{R})$ we denote the Sobolev space of weakly differentiable square integrable functions. 
Given $\phi, \psi \in L^2(-d,0;\mathbb{R}))$, by $\langle\cdot, \cdot\rangle$ we denote the usual $L^2$ scalar product.
\\
If functions  $a,b\ge 0$ satisfy the inequality $a \le C(A) b$ with a constant $C(A)>0$ depending on the expression $A$, we write $a \lesssim_A b$.\end{notation}

\section{Setup and statement of the main result}
\label{Section_Mathematical_tools}
We work in the framework of a Black and Scholes complete market model. On the filtered probability space $(\Omega, \mathcal F, \mathbb F, \mathbb P)$ we consider the $\mathbb F$-adapted vector valued process $(S_0,S)$, representing the price evolution of a money market account, $S_0$,
and $n$ risky assets, $S=(S_1,...,S_n)^\top$, whose dynamics is the following 
\begin{equation}
\label{DYNAMIC_MARKET}
\begin{cases}
\text{d}S_0(t)= S_0(t) r  \text{d}t,\\
\text{d}S(t) =\text{diag}(S(t)) \left(\mu \text{d}t + \sigma \text{d}Z(t)\right),
\\
S_0(0)=1,
\\
S(0)\in  {\mathbb R}^n_{+}.
\end{cases}
\end{equation}
Here $Z$ is an $n$-dimensional Brownian motion and we assume that $\mathbb F:=\{\mathcal F_t\}_{t\geq0}$ is the filtration
generated by $Z$, and enlarged with the
$\mathbb P$-null sets. $\mu \in \mathbb R^n$, and the matrix $\sigma \in  \mathbb R^{n \times n}$ is invertible.

We are interested in the valuation of a payment stream represented by the $\mathbb F$-adapted process $y$. The payment stream can be thought of as capturing the mark-to-market process of a trading account, 
the flow of profits and losses from a trading strategy, the collateral flows arising from an over-the-counter derivative transaction, or the labor income received
by an agent over time. 
We assume the  payment stream $y$ to obey the following stochastic differential equation with delay in both the drift and the diffusion terms:
 \begin{equation}\label{DYN_LABOR_INCOME_DELAY_I}
\begin{cases}
\text{d}y(t) = &\left[ y(t) \mu_y+\int_{-d}^0 y(t+s)  \phi(\text{d}s) \right] \text{d}t \\
 &+\left[ y(t)\sigma_y + \begin{pmatrix}
 \int_{-d}^0 y(t+s)  \varphi_1(\text{d}s) \\
 \vdots
 \\
 \int_{-d}^0y(t+s) \varphi_n(\text{d}s)
\end{pmatrix}    \right] \cdot \text{d}Z(t),
\\
y(0)=& x_0, 
\\
y(s)=& x_1(s )  \quad \mbox{  for $s \in [-d,0)$}.
\end{cases}
 \end{equation}
Here $\mu_y \in \mathbb R$, $\sigma_y \in \mathbb R^n$,
  $\phi,\varphi_i$, for $i=1,\dots,n$, 
 are signed measures of bounded variation on $\left[-d,0\right]$
 and 
 $x_0 \in \mathbb R$, $x_1 \in  L^2 \big( [-d,0]; \mathbb R\big) $.
 The income stream $y$
 provides a simple, tractable example of income dynamics adjusting slowly to financial market shocks; the measures $\phi,\varphi_1,...\varphi_n$ can be thought as representing the impact that the past has on the present value of the payment stream $y$. 
 
Existence of a unique solution of \eqref{DYN_LABOR_INCOME_DELAY_I} is ensured by the following result.
 \begin{proposition}
 \label{well_pose}
For any given initial datum $(x_0,x_1)\in \mathbb{R} \times L^2(-d,0;\mathbb{R})$ equation \eqref{DYN_LABOR_INCOME_DELAY_I} admits a unique strong (in the probabilistic sense) solution $y\in L^2(\Omega \times [0,T] )$, for all $T>0$, with $\mathbb{P}$-a.s. continuous paths. 
\end{proposition}
Proof of Proposition \eqref{well_pose} is given in Appendix \ref{app0}.

Since the market is complete and the stream process $y$ is spanned by the stock $S$, a certainty equivalent value for the stream of future (stochastic) payments can be found. It is in fact well known that in a complete market, the value of future uncertain cash flows is obtained by discounting their expected value under the unique risk neutral measure. This value is 
\begin{eqnarray}\label{HUMAN_CAPITAL}
H\left(t_0\right):=\xi(t_0)^{-1} \mathbb E\left[ \int_{t_0}^{+ \infty} \xi(s) y(s) \text{d}s \Big\vert \mathcal F_{t_0}\right],
\end{eqnarray}
where the process $\xi$ represents the unique stochastic discount factor whose time evolution is described (see e.g. \cite{Duf2001}) by the following stochastic differential equation:
\begin{equation}\label{DYN_STATE_PRICE_DENSITY}
\left\{\begin{array}{ll}
\text{d} \xi (t)& = - \xi(t)r  \text{d}t  -\xi(t) \kappa \cdot \text{d}Z(t)\\
\xi(0)&=1.
\end{array}\right.
\end{equation}
The constant $\kappa$ that appears in the above equation, represent the market price of risk. It is known that in a Black and Scholes market model 
\begin{equation}\label{DEF_KAPPA}
\kappa= (\sigma^\top)^{-1} (\mu- r \mathbf 1),
\end{equation}
where by $\mathbf 1 = (1,\dots, 1)^\top$ we mean the unitary vector in $\mathbb R^n$.

Our aim is to provide an explicit formula for the quantity $H(t_0)$ given in \eqref{HUMAN_CAPITAL}. The challenging aspect of the problem lies  in the fact that we consider a payment stream whose dynamics is path dependent. This makes the problem considerably different from the ones studied in the literature and much harder to prove.

Now we have all the ingredients to state the main result of the paper. Throughout the paper we will let the following assumption be in force

\begin{hypothesis}
\label{HYP_POSITIVITY_PHI_K}
\begin{itemize}
\item [(i)] The measure $\Phi$ on $[-d,0]$ defined as 
\begin{equation}
\label{DEF_PHI}
\Phi(\cdot):=\phi(\cdot)-(\varphi_1(\cdot),...,\varphi_n(\cdot)) \cdot \kappa
\end{equation}
is a signed measure of bounded variation.
\item [(ii)] 
\begin{equation}
\label{EQ_POSITIVITY_K}
r-\mu_0+\sigma_y \cdot \kappa - \int_{-d}^0 e^{r\tau}|\Phi|({\rm d} \tau)>0.
\end{equation}
\end{itemize}
\end{hypothesis}

Assumption \ref{HYP_POSITIVITY_PHI_K} will be explained in details in Section \ref{Explanation}.

\begin{theorem}
\label{TEO_EXPLICIT_FORM_HUMAN_CAPITAL}
Under Assumption \ref{HYP_POSITIVITY_PHI_K}, for any $t_0\ge 0$, the quantity $H(t_0)$ defined in \eqref{HUMAN_CAPITAL} has the following explicit form 
\begin{equation}\label{SOLUTION}
  H(t_0)=
  \frac{1}{K} \left( y(t_0)+ \int_{-d}^0  G(s) y(t_0+s)\, {\rm d}s\right),\quad\mathbb P-a.s.,
\end{equation}
where $y(t_0)$ denotes the solution at time $t_0$ of equation \eqref{DYN_LABOR_INCOME_DELAY_I},
\begin{equation*}
\label{def_K}
K:=r-\mu_0+\sigma_y \cdot \kappa - \int_{-d}^0 e^{r\tau}\Phi({\rm d} \tau), 
\end{equation*}
and $G$ is given by
\begin{equation*}
\label{DEF_G}
G(s) : = \int_{-d}^s   e^{-r(s-\tau)} \Phi ({\rm d}\tau).
\end{equation*}
\end{theorem}
In expression \eqref{SOLUTION}, we can recognize an annuity factor, $K^{-1}$, multiplying
a term representing current value of $y$, 
and a term representing the current market value of the past trajectory of $y$ over
the time window $(t_0-d,t_0)$. The `market value
of the past' trades off the returns on the payment stream against its exposure to financial risk,
as can be seen from expression \eqref{DEF_PHI}.
When the delay terms in the drift and volatility coefficients vanish, the valuation 
of the payment stream reduces to $K^{-1}y(t_0)$.

\begin{remark}\label{remark:bounded_horizon}
The setup described above can be extended to the case of payments over a bounded horizon in some interesting situations.
\\
When the payment stream is received until an exogenous Poisson stopping time $\tau_\delta$ (representing for example death or irreversible unemployment, in the case payment stream $y$ represents the labor income), expression \eqref{SOLUTION} still applies, provided discounting is carried out at rate $r+\delta$ instead of $r$,
where $\delta>0$ represents the Poisson parameter. This extension has been already considered in \cite{BGP}, \cite{BiaGozZan} and \cite{DGZZ}.
\\
The case in which payments are received until a finite (deterministic)  time (representing for example permanent retirement from the labor market, in the case payment stream $y$ represents the labor income) has been addressed in \cite{BCGZ} (see Proposition 3.4). 
\end{remark}

\section{Proof of the result}
\label{proof_section}

Within this Section we assume that \eqref{DYN_LABOR_INCOME_DELAY_I} has a unique continuous $\F$-adapted solution $y$, as ensured by Proposition \ref{well_pose}.
The proof of Theorem \ref{TEO_EXPLICIT_FORM_HUMAN_CAPITAL} can be divided in the following steps:
\begin{itemize}
\item incorporate the discount factor $\xi$ in the equivalent risk-neural probability measure $\tilde{\mathbb{P}}$ and rewrite the dynamics of $y$ under $\tilde{\mathbb{P}}$. Derive the deterministic delayed equation satisfies by the quantity $\tilde{\mathbb{E}}[y(t)|\mathcal{F}_{t_0}]$.  (Subsection \ref{subsec:equiv_P}).
\item Rewrite the delayed equation for $\tilde{\mathbb{E}}[y(t)|\mathcal{F}_{t_0}]$ as a deterministic evolution equation taking values in a suitable Hilbert space that incorporate in its structure the past and the present. We will appeal to the so-called \textit{product-space framework} for path-dependent equations (Subsection \ref{subsec:Hilbert}).
\item Exploit suitable spectral properties of the operator that appears in the above mentioned infinite-dimensional formulation of the problem, to obtain expression \eqref{SOLUTION} for $H(t_0)$ (Subsections \ref{subsec:spectral} and \ref{subsec:formula_HC}). 
\item Clarify the relation between the used spectral properties and our Assumption \ref{HYP_POSITIVITY_PHI_K} (Subsection \ref{Explanation}).
\end{itemize}

The above first three steps will lead to Proposition \ref{step1}, whereas the last step will be formalized in Lemmas \ref{lem_k} and \ref{spectral_bounds}. Theorem \ref{TEO_EXPLICIT_FORM_HUMAN_CAPITAL} will then follow as an immediate consequence of these results.

For readability, the proof of some technical Lemmas in the following Sections will be postponed to the Appendix.

\subsection{Equivalent probability measure}
\label{subsec:equiv_P}

We find more convenient to incorporate the discount factor $\xi$ in the equivalent probability measure $\tilde{\mathbb{P}}$. It is in fact well known that, in a complete market model, pricing formulas can be obtained, in a equivalent way, working under the (unique) risk-neutral measure $\tilde{\mathbb{P}}$ or working under the objective measure $\mathbb{P}$ but discounting by the state price process $\xi$.

We start by considering the equivalent probability measure $\tilde \P_s$ on $\mathcal{F}_s$ such that 
\begin{footnote}
{Recall that the state price density $\xi$ characterizes the Radon-Nikodym derivative that defines the change of probability measure from objective $\mathbb{P}$ to risk-neutral measure $\tilde{\mathbb{P}}$ via the relation: $\xi(s)=e^{-rs}\rho(s)= e^{-rs}\frac{{\rm d}\tilde{\mathbb{P}}}{{\rm d}\mathbb{P}}(s)$.}
\end{footnote}

\begin{equation}
\label{Ptilde}
  \frac{\text{d}\tilde{\mathbb  P}_s}{\text{d} \mathbb P}=\exp\left(  -\frac{1}{2} |\kappa|^2  s   - \kappa\cdot Z(s)\right)     =  e^{rs}  \xi(s).
\end{equation}
By \cite[Lemma 3.5.3]{KARATZAS_SHREVE_91} we can write  
 \[  
 \mathbb E   \left[ \xi(s)  y(s) \mid \mathcal F_{t_0}  \right] = 
\xi(t_0) e^{-r(s-t_0)}  \tilde{\mathbb E}_s \left[  y(s) \mid \mathcal F_{t_0}  \right]  ,\]
and thus
\begin{footnote}{
Recall (see \eqref{HUMAN_CAPITAL}) that our aim is to evaluate $\mathbb E \left[ \int_{t_0}^{+\infty}   \xi(s)  y(s)  \text{d}s \mid \mathcal F_{t_0}    \right]$. We will prove that  $\int_{t_0}^{+\infty}  \mathbb E \left[ \xi(s)  y(s)\mid \mathcal F_{t_0}   \right]\text{d}s$ is equal to the r.h.s. of \eqref{SOLUTION} and then justify the equality  $\mathbb E \left[ \int_{t_0}^{+\infty}   \xi(s)  y(s)  \text{d}s \mid \mathcal F_{t_0}    \right]=\int_{t_0}^{+\infty}  \mathbb E \left[ \xi(s)  y(s)\mid \mathcal F_{t_0}   \right]\text{d}s$.}\end{footnote}
  \begin{equation}
    \label{EXPRESSION_II}
    \begin{aligned}
 \int_{t_0}^{+\infty}  \mathbb E \left[ \xi(s)  y(s)\mid \mathcal F_{t_0}   \right]\text{d}s
    =
    \xi(t_0) e^{rt_0} \int_{t_0}^{+\infty}  e^{-rs}  \tilde{\mathbb E}_s \left[  y(s)\mid \mathcal F_{t_0}   \right] \text{d}s.
  \end{aligned}
\end{equation}

The idea is to now understand what kind of differential equation the quantity $\tilde{\mathbb E}\left[  y(s)\mid \mathcal F_{t_0}   \right]=\tilde\E_s\left[  y(s)\mid \mathcal F_{t_0}   \right]$ satisfies.
Let $\tilde{\mathbb P}$ the measure such that $\left.\tilde{\mathbb P}\right|_{\mathcal F_s}=\tilde{\mathbb P}(s)$ for all $s\geq 0$.
By the Girsanov Theorem the process
$\tilde{Z}(t) = Z(t) + \kappa t$
is an $n$-dimensional Brownian motion under $ \tilde{\mathbb  P}$.
The dynamics of $y$ under $ \tilde{\mathbb  P}$ is then

\begin{eqnarray}
\label{y_tilde_P}
dy(s) =
 & \big[  (\mu_y - \sigma_y \cdot \kappa) y(s) + \int_{-d}^0  y(s+\tau) \,\Phi(\text{d}\tau)   \big] \text{d}s+\left[ y(t)\sigma_y  + \begin{pmatrix}
 \int_{-d}^0 y(s+\tau) \varphi_1(\text{d}\tau) \\
 \vdots
 \\
 \int_{-d}^0 y(s+\tau) \varphi_n(\text{d}\tau)
\end{pmatrix}    \right] \cdot  \text{d} \tilde{Z}(s),
  \end{eqnarray}
  where $\Phi$ is defined in \eqref{DEF_PHI}.
  Integrating between $t_0$ and $t$ we obtain
  \begin{align}
  \begin{split}
y(t)  = y(t_0) & + \int_{t_0}^t     (\mu_y - \sigma_y \cdot \kappa) y(s) \text{d}s
  +  \int_{t_0}^t     \int_{-d}^0 y(s+\tau)\Phi(\text{d}\tau)  \text{d}s\\
&  +   \int_{t_0}^t \left[ y(s)\sigma_y  + \begin{pmatrix}
 \int_{-d}^0 y(s+\tau) \varphi_1(\text{d}\tau)  \\
 \vdots
 \\
 \int_{-d}^0 y(s+\tau) \varphi_n(\text{d}\tau)
\end{pmatrix}    \right]  \cdot \text{d}\tilde{Z}(s)\,,
\end{split}
  \end{align}
and therefore, by taking the conditional expected value on both sides, we get
  \begin{align}\label{FORMULA_DIM_TEO_H_0_INTEGRAL X_XI}
  \begin{split}
   \tilde{\mathbb E} \left[  y(t)\mid \mathcal F_{t_0}   \right]
 =& y(t_0) +   (\mu_y - \sigma_y\cdot \kappa)   \tilde{\mathbb E} \left[   \int_{t_0}^t    y(s)  \text{d}s  \mid \mathcal F_{t_0} \right]
 \\
  &+     \tilde{\mathbb E} \left[   \int_{t_0}^t     \int_{-d}^0  y(s+ \tau) \Phi(\text{d} \tau)   \text{d}s \mid \mathcal F_{t_0}  \right] \\
  &+    \tilde{\mathbb E} \left[  \int_{t_0}^t \left[ y(s)\sigma_y  + \begin{pmatrix}
 \int_{-d}^0  y(s+\tau)\varphi_1(\text{d}\tau) \\
 \vdots
 \\
 \int_{-d}^0  y(s+\tau) \varphi_n(\text{d}\tau)
\end{pmatrix}    \right] \cdot \text{d}\tilde{Z}(s) \mid  \mathcal F_{t_0}  \right].
 \end{split}
  \end{align}

  The following  Lemma guarantees that
the stochastic integral with respect 
to $\tilde{Z} $ is a martingale, and has zero mean. 
The proof is provided in Appendix \ref{app1}
  
   \begin{lemma}\label{LEMMA_STOCHASTIC_INTEGRAL}
  It holds that
  \begin{eqnarray*}
   \tilde{\mathbb{E}}\left[  \int_{t_0}^t   \left \Vert y(s)\sigma_y  + \begin{pmatrix}
 \int_{-d}^0  y(s+\tau)\varphi_1({\rm d}\tau)\\
 \vdots 
 \\
 \int_{-d}^0 y(s+\tau) \varphi_n({\rm d}\tau)
\end{pmatrix}  \right\Vert ^2\, {\rm d} s  \right] < +\infty\,.
  \end{eqnarray*}
  \end{lemma}

We thus obtain that  
\begin{equation*}
 \tilde{\mathbb E} \left[  \int_{t_0}^t \left[ y(s)\sigma_y  + \begin{pmatrix}
 \int_{-d}^0  y(s+\tau)\varphi_1(\text{d}\tau) \\
 \vdots
 \\
 \int_{-d}^0  y(s+\tau) \varphi_n(\text{d}\tau)
\end{pmatrix}    \right] \cdot \text{d}\tilde{Z}(s) \mid  \mathcal F_{t_0}  \right]=0,
\end{equation*}
 and,
by definition of conditional mean and by Fubini's Theorem, the expression in
 \eqref{FORMULA_DIM_TEO_H_0_INTEGRAL X_XI} reduces to
  \begin{align}\label{FORMULA_DIM_TEO_H_0_INTEGRAL X_XI_II}
  \begin{split}
   \tilde{\mathbb E} \left[  y(t)\mid \mathcal F_{t_0}   \right]
= \ & y(t_0) +   (\mu_y - \sigma_y\cdot \kappa)    \int_{t_0}^t  \tilde{\mathbb E}  \left[   y(s)  \mid \mathcal F_{t_0} \right]\, {\rm d}s
 \\
  &+       \int_{t_0}^t     \int_{-d}^0 \tilde{\mathbb E} \left[ y(s+ \tau) \mid \mathcal F_{t_0}  \right]\,  \Phi(\text{d} \tau)   \text{d}s. 
   \end{split}
  \end{align}
  Therefore, defining 
  \begin{equation}
  \label{def_M0}
  M_{t_0}(t):=\tilde{\mathbb{E}}\left[y(t)|\mathcal{F}_{t_0}\right],
  \end{equation}
   we have that $M_{t_0}$ satisfies for $t\geq t_0$ the equation (with random initial conditions)  
   \begin{equation}\label{eq1_I}
\begin{cases}
{\rm d}M_{t_0}= [(\mu_y - \sigma_y \cdot \kappa)   M_{t_0}(t)+\int_{-d}^0   M_{t_0}(t+s)\, \Phi(\text{d} s)]\, {\rm d}t,
\\
M_{t_0}(t_0)=y(t_0),& 
\\
M_{t_0}(t_0+s)=y(t_0+s), \ \quad \qquad s \in[-d,0). 
\end{cases}
\end{equation}
Existence of a unique solution of the above system is guaranteed by the following generalization of   \cite[Part II, Chapter 4, Theorem 3.2]{BENSOUSSAN_DAPRATO_DELFOUR_MITTER} to random initial conditions.
\begin{lemma}
  Given any fixed $\mathcal{F}_{t_0}$-measurable $\mathbb{R} \times L^2([-d,0];\mathbb{R})$-valued random variable $m=(m_0,m_1)$, the Cauchy problem 
  \begin{equation}
  \label{new}
\begin{cases}
{\rm d} m(t_0;t)= [(\mu_y - \sigma_y \cdot \kappa)   m(t_0;t)+\int_{-d}^0   m(t_0;t+s)\, \Phi({\rm d} s)]\, {\rm d}t,
\\
m(t_0;t_0)=m_0,& 
\\
m(t_0;t_0+s)=m_1(s), \ \quad \qquad s \in[-d,0). 
\end{cases}
\end{equation}
   admits a unique absolutely continuous solution. Moreover, system \eqref{new} is equivalent to \eqref{eq1_I} when we choose $(m_0,m_1)=(y(t_0),y(t_0+\cdot))$.
   \end{lemma}
\subsection{Reformulation of the problem in an infinite-dimensional framework}
\label{subsec:Hilbert}

We now reformulate the differential equation with delay \eqref{eq1_I} as an evolution equation with values in the so called Delfour-Mitter Hilbert space, defined as \begin{equation*}
\mathcal{H}:=\R\times L^2(-d,0;\R),
\end{equation*}
whose elements are denoted as $ x=(x_0,x_1)$. $\mathcal{H}$ is a Hilbert space when endowed with the inner product $\langle (x_0,x_1),(y_0,y_1)\rangle_{\mathcal{H}}=x_0y_0+\langle x_1,y_1\rangle$, the latter being the usual inner product of $L^2(-d,0;\R)$.

We define the operator $A:\mathcal{D}(A) \subset \mathcal{H} \rightarrow \mathcal{H}$ as
\begin{equation*}
\mathcal D\left(A\right):=\left\{(x_0,x_1) \in \mathcal H: x_1(\cdot) \in W^{1,2}\big( [-d, 0]; \mathbb R \big),\, x_0 = x_1(0)\right\},
\end{equation*}
\begin{equation}
\label{A}
A(x_0,x_1):=\Big(
 (\mu_0 - \sigma_0\cdot  \kappa) x_0+\int_{-d}^0  x_1(s)\Phi(\text{d}s),
\frac{\text{d}}{\text{d}s} x_1
\Big),
\end{equation}
with $\Phi$ defined in \eqref{DEF_PHI}.

We can then reformulate equation \eqref{eq1_I} as an evolution equation in $\mathcal{H}$. 

Consider, again for any fixed $\mathcal{F}_{t_0}$-measurable $\mathcal{H}$-valued random variable $\mathbf{m}=\left(m_0,m_1\right)$, the $\mathcal{H}$-valued process $\mathbf{M}(t_0;\cdot)$ that is the solution on $[t_0,+\infty)$ of
\begin{equation}\label{INFINITE_DIMENSIONAL_STATE_EQUATION}
\begin{cases}
{\rm d} \mathbf{M}(t_0;t) = A \mathbf{M}(t_0;t) {\rm d}t,\\
\mathbf{M}(t_0;t_0) =\mathbf{m}.
\end{cases}
\end{equation}
We collect in the following Proposition some useful results about the above equation (for more details see e.g. \cite[Appendix A]{DAPRATO_ZABCZYK_RED_BOOK}).
\begin{proposition}
  \label{prop_semigroup}
  \begin{enumerate}[label=$(\roman{*})$]
  \item\label{item:1} The operator $A$ generates a strongly continuous semigroup $\left\{S(t)\right\}_{t\geq 0}$ in $\mathcal{H}$.
  \item $S(t)$ is a compact operator for every $t\geq d$.
  \item For every $\mathcal{F}_{t_0}$-measurable $\mathcal{H}$-valued random variable $m$ the process
\begin{equation}
  \label{semigroup_1}
S(t-t_0)\mathbf{m} ;
\end{equation}
is the unique weak (in distributional sense) solution of (\ref{INFINITE_DIMENSIONAL_STATE_EQUATION}); in particular
\begin{equation}
  \label{semigroup_2}
 \mathbf{M}(t_0;t)=\mathbf{M}(0;t-t_0)\ .
\end{equation}
\item The Cauchy problem (\ref{INFINITE_DIMENSIONAL_STATE_EQUATION}) is equivalent to (\ref{new}).
\end{enumerate}
\end{proposition}
\begin{proof}
  \begin{enumerate}[label=$(\roman{*})$]
\item 
The operator $A$ can be written in the form
\begin{equation}\label{eq_gen}
A\left(x_0,x_1\right)=\left(\int_{-d}^0 x_1(\theta)a(\text{d}\theta), \frac{\text{d}}{\text{d}s}x_1\right)\,,
\end{equation}
where
\[a(\text{d}\theta)=\mu_y\delta_0(\text{d}\theta)+\Phi(\text{d}\theta)\,,\]
and $\delta_0$ is the delta-Dirac measure at zero.
The measure $a$ defines a finite measure on $\left[-d,0\right]$. The result is thus an immediate consequence of \cite[Proposition A.27]{DAPRATO_ZABCZYK_RED_BOOK}.
\item See e.g. \cite[Chapter 7, Lemma 1.2]{HALE_VERDUYN_LUNEL_BOOK}. 
\item Existence and uniqueness of a weak solution given by (\ref{semigroup_1}) for deterministic $\mathbf{m}$ is a classical result (see \cite[Proposition A.5]{DAPRATO_ZABCZYK_RED_BOOK}). 
One can then easily generalize the result to random $\mathbf{m}$. Property (\ref{semigroup_2}) follows from uniqueness of the solution.

\item\label{it:equiv} If $m(t_0;\cdot)$ is the unique solution to (\ref{new}) then the $\mathcal{H}$-valued process $\left(m(t_0;t),m(t_0;t+\cdot)\right)_{t\geq t_0}$ solves (\ref{INFINITE_DIMENSIONAL_STATE_EQUATION}) by \cite[Part II, Chapter 4, Theorem 4.3]{BENSOUSSAN_DAPRATO_DELFOUR_MITTER}. Since also the latter has a unique solution, its first component must be the solution to (\ref{INFINITE_DIMENSIONAL_STATE_EQUATION}).
\end{enumerate}
\end{proof}

As an immediate consequence of the above result we obtain the desired equivalence between equations \eqref{INFINITE_DIMENSIONAL_STATE_EQUATION} and (\ref{eq1_I}).\begin{corollary}
\label{item:5} 
Let $y$ be a solution of (\ref{DYN_LABOR_INCOME_DELAY_I}) on $[0,t_0]$; when choosing $\mathbf{m}$ as $(m_0,m_1)= (y(t_0), y(t_0+\cdot))$, \eqref{INFINITE_DIMENSIONAL_STATE_EQUATION} is equivalent to (\ref{eq1_I}) and in this case we have
\begin{equation*}
\mathbf{M}(t_0;t)=S(t-t_0)\mathbf{m}
=\left(m(t_0;t),m(t_0;t+\cdot)\right)
=\left(M_{t_0}(t),\left\{M_{t_0}(t+s)\right\}_{s\in[-d,0]}\\\right)\ .
\end{equation*}
\end{corollary}

From now on we thus will work with formulation \eqref{INFINITE_DIMENSIONAL_STATE_EQUATION}. The spectral properties of the operator $A$, that appears in this infinite-dimensional formulation, will be crucial to prove our result. We devote the next Section to the analysis of these properties.

\subsection{Spectral properties of $A$}
\label{subsec:spectral}
In the present Section we collect some technical results concerning the spectral properties of the operator $A$. Proof of Theorem \ref{TEO_EXPLICIT_FORM_HUMAN_CAPITAL} is based on the Lemmas presented here. The technical proofs are postponed to the Appendix.

\begin{lemma}
  \label{lemma_speck}
The spectrum of the operator $A$ is given by
\begin{equation*}
\sigma(A)=\{\lambda \in \mathbb{C}:K(\lambda)=0\},
\end{equation*}
where
\begin{equation}
\label{DEF_K_LAMBDA}
K(\lambda):=\lambda -  (\mu_y- \sigma_y \cdot \kappa) - \int_{-d}^0   e^{\lambda \tau}  \Phi ({\rm d}\tau)\,,\quad\lambda\in\mathbb C\,.
\end{equation}
The spectrum $\sigma\left(A\right)$ is a countable set and every $\lambda\in\sigma\left(A\right)$ is an isolated eigenvalue of finite multiplicity. 

The spectral bound of $A$ is 
\begin{equation}\label{l0}
\lambda_0=\sup\left\{\mathrm{Re}\,\lambda:\, K(\lambda)=0\right\}.
\end{equation}
\end{lemma}
\begin{proof}
See \cite[Chapter 7, Lemma 2.1 and Theorem 4.2]{HALE_VERDUYN_LUNEL_BOOK}\end{proof}
We ca explicitly compute the resolvent operator of $A$.

\begin{lemma}\label{LEMMA_RESOLVENT_SET_RESOLVENT_BAR_A}
Let $\rho(A)$ denote the resolvent set of $A$ and let $\lambda\in\mathbb R\cap\rho\left(A\right)$.
The resolvent operator of $A$ at $\lambda$, denoted by $R(\lambda,A)$ is given by
\begin{equation}
R(\lambda, A)  \left(m_0,
m_1\right)=
 \left(
u_0,
u_1\right)
\end{equation}
with

\begin{multline}
\begin{split}
u_0&=\frac{1}{K(\lambda)}\left[m_0+ \int_{-d}^0\int_{-d}^{s}e^{-\lambda (s-\tau)}\Phi({\rm d}\tau) \, m_1(s)  {\rm d} s\right] ,\\
u_1(s)&=
u_0e^{\lambda s}
+\int_s^0e^{-\lambda(\tau-s)}m_1(\tau)\, {\rm d}\tau.
\end{split}
\end{multline}

\end{lemma}

\begin{proof}
See Appendix \ref{app2}.
\end{proof}

\begin{lemma}
\label{lem_laplace}
For every real $\lambda$ such that $\lambda>\lambda_0$ and every $\mathbf{m}=(m_0,m_1)\in\mathcal{H}$ we have
    \begin{equation}
      \label{LAPLACE_SEMIGROUP}
  \int_0^{+\infty}e^{-\lambda t}S(t)  \mathbf{m}\, {\rm d}t=R(\lambda,A)\mathbf{m}.
\end{equation}
\end{lemma}
\begin{proof}
 Identity (\ref{LAPLACE_SEMIGROUP}) is well known to hold for all real $\lambda$ larger than the type of $S(t)$. Since $S(t)$ is compact for every $t\geq d$, its type is actually equal to its spectral radius $\lambda_0$. For a reference see e.g. \cite[Part II, Chapter 1, Corollary 2.5]{BENSOUSSAN_DAPRATO_DELFOUR_MITTER}.
\end{proof}

\subsection{Deriving the explicit formula for $H$}
\label{subsec:formula_HC}
In this Section we exploit the results derived in the above Sections to prove the following.
\begin{proposition}
\label{step1}
Assume $r>\lambda_0$, then for any $t_0\ge 0$, the quantity $H(t_0)$ defined in \eqref{HUMAN_CAPITAL} has the following explicit form 
\begin{equation*}
  H(t_0)=
  \frac{1}{K} \left( y(t_0)+ \int_{-d}^0  G(s) y(t_0+s)\, \text{d}s\right),\quad\mathbb P-a.s.,
\end{equation*}
where $y(t_0)$ denotes the solution at time $t_0$ of equation \eqref{DYN_LABOR_INCOME_DELAY_I},
\begin{equation}
\label{def_K}
K:=r-\mu_0+\sigma_y \cdot \kappa - \int_{-d}^0 e^{r\tau}\Phi({\rm d} \tau), 
\end{equation}
and $G$ is given by
\begin{equation}
\label{DEF_G}
G(s) : = \int_{-d}^s   e^{-r(s-\tau)} \Phi ({\rm d}\tau).
\end{equation}
\end{proposition}
\begin{remark}
 Notice that the statement of the above result is the same of Theorem \ref{TEO_EXPLICIT_FORM_HUMAN_CAPITAL}, but the assumptions here are different: we assume $r>\lambda_0$ instead of Assumption \ref{HYP_POSITIVITY_PHI_K}. An explanation of why we do actually consider  Assumption \ref{HYP_POSITIVITY_PHI_K} in Theorem \ref{TEO_EXPLICIT_FORM_HUMAN_CAPITAL}, will be provided in the next Section.
 \end{remark}
 \begin{proof} 
Let $\mathbf{m}=(m_0,m_1)= (y(t_0), y(t_0+\cdot))$. We denote here by $\Pi$ the projection on the first (finite-dimensional) component of $\mathcal{H}$, i.e. $ \Pi[ \mathbf{m}] =\Pi[(m_0,m_1)]=m_0$.

Starting from \eqref{EXPRESSION_II}, we have
\begin{align}
\label{final}
 \frac{1}{\xi(t_0)}\int_{t_0}^{+\infty}  \mathbb E& \left[ \xi(s)  y(s)\mid \mathcal F_{t_0}   \right]\text{d}s
   =e^{rt_0}\int_{t_0}^{\infty}e^{-rs}\tilde{\mathbb{E}}\left[y(s)|\mathcal{F}_{t_0}\right]\, {\rm d}s
&\text{(by (\ref{EXPRESSION_II}))}
\notag \\
  &=e^{rt_0}\int_{t_0}^{\infty}e^{-rs}M_{t_0}(s)\, {\rm d}s &\text{(by \eqref{def_M0}}
 \notag \\  
  &=e^{rt_0}\int_{t_0}^{\infty}e^{-rs}\Pi\left[\mathbf{M}(t_0;s)\right]\, {\rm d}s&\text{(by Corollary \ref{item:5})}
 \notag \\
 &=e^{rt_0}\int_{0}^{\infty}e^{-rt_0}e^{-rs}\Pi\left[\mathbf{M}(0;s)\right]\, {\rm d}s&\text{(by (\ref{semigroup_2}))}
\notag \\  
 &= \int_0^{\infty}e^{-rs}\Pi\left[S(s)\mathbf{m}\right]\,{\rm d}s&\text{(by (\ref{semigroup_1}))}
 \notag \\
  &=\Pi\left[R(r,A)\mathbf{m}\right]&\text{(by Lemma \ref{lem_laplace}, since $r>\lambda_0$)}
  \notag \\
 &= \frac{1}{K(r)}\left[y(t_0)+\int_{-d}^0\int_{-d}^{s}   e^{-r(s- \tau)} \Phi(\text{d}\tau) \, y(t_0+s) \text{d} s\right]
&\text{(by Lemma \ref{LEMMA_RESOLVENT_SET_RESOLVENT_BAR_A}).}
 \end{align}
 From the above equalities we infer, in particular, the $\mathbb{P}$-integrability of $ \int_{t_0}^{+\infty}   \mathbb E [ \xi(s)  y(s)   \mid \mathcal F_{t_0}    ]\,{\rm d}s$, which justifies 
 the equality
 \begin{equation}
 \label{final2}
 \mathbb E \left[ \int_{t_0}^{+\infty}   \xi(s)  y(s)  \text{d}s \mid \mathcal F_{t_0}    \right]=\int_{t_0}^{+\infty}  \mathbb E \left[ \xi(s)  y(s)\mid \mathcal F_{t_0}   \right]\text{d}s.
 \end{equation} 
 In fact, by the characteristic property of the conditional mean,
 and Fubini's Theorem we have that, 
for any $F \in  \mathcal F_{t_0}$ 
\begin{align*}
\int_F   \int_{t_0}^{+\infty}   \mathbb E [ \xi(s)  y(s)   \mid \mathcal F_{t_0}    ]\,{\rm d}s \,  {\rm d}\mathbb P 
&=  \int_{t_0}^{+\infty} \int_F  \mathbb E [ \xi(s)  y(s)   \mid \mathcal F_{t_0}    ] {\rm d}\mathbb P \, {\rm d}s
\\
& =  \int_{t_0}^{+\infty}   \int_F    \xi(s)  y(s) \, {\rm d} \mathbb P \, {\rm d}s
 =   \int_F     \int_{t_0}^{+\infty}   \xi(s)  y(s) \, {\rm d}s  \, {\rm d} \mathbb P  \\
&=\int_F      \mathbb E \left[\int_{t_0}^{+\infty}   \xi(s)  y(s)\, {\rm d}s \mid \mathcal F_{t_0}     \right]    \,   {\rm d} \mathbb P.
\end{align*}

Defining now $K:=K(r)$ and recalling \eqref{HUMAN_CAPITAL}, \eqref{def_K} and \eqref{DEF_G}, by \eqref{final} and \eqref{final2}, the result immediately follows.
\end{proof}
\subsection{Motivations for Assumption \ref{HYP_POSITIVITY_PHI_K}}
\label{Explanation}
In Proposition \ref{step1} we proved our main result under the Assumption $r>\lambda_0$. This requirement is difficult to check in practice, since it requires an explicit computation of the spectral bound $\lambda_0$. In the present Section we therefore look for some sufficient conditions easier to check. 


Set for $\lambda \in \mathbb{C}$
\begin{equation}
\label{tilde_K}
\widetilde K(\lambda):=\lambda -  (\mu_y- \sigma_y \cdot \kappa) - \int_{-d}^0   e^{\lambda \tau}  |\Phi |(\text{d}\tau),
\end{equation}
where by $|\Phi|$ we denote the total variation measure of $\Phi$.
Set 
\begin{equation}\label{l1}
\widetilde{\lambda_0}=\sup\left\{\mathrm{Re}\,\lambda:\, \widetilde K(\lambda)=0\right\}.
\end{equation}
We note that $\widetilde{\lambda_0}$ is the spectral radius of the operator $\widetilde{A}: \mathcal{D}(\widetilde{A})\subset \mathcal{H} \rightarrow \mathcal{H}$ defined as follows:
\begin{equation*}
  \begin{gathered}
\mathcal{D}(\widetilde{A}):= \left\{ \left(x_0,x_1\right) \in \mathcal{H}: x_1\in W^{1,2}(-d,0;\mathbb{R}), \ x_1(0)=x_0\right\},
\\
\widetilde{A}\left(x_0,x_1\right):= \left((\mu_y-\sigma_y\cdot \kappa)x_0+ \int_{-d}^0 x_1(s)|\Phi|({\rm d}s), \frac{{\rm d}}{{\rm d}s} x_1\right).
\end{gathered}
\end{equation*}

\begin{lemma}
\label{lem_k}
 The function $\widetilde K$, restricted to the real numbers, is strictly increasing and the spectral bound $\widetilde{\lambda_0}$ is the only real root of the equation
$\widetilde K(\xi)=0$. 
In particular, 
\begin{equation}
\label{iff_condition}
\widetilde K(r)>0 \ \iff   r>\widetilde{\lambda_0}.
\end{equation}
\end{lemma}
\begin{proof}
See Appendix \ref{app3}.
\end{proof}

Recall the definition of $K$ given in \eqref{DEF_K_LAMBDA} and the definition of the spectral bound of $A$, $\lambda_0$, given 
\eqref{l0}.

\begin{lemma}
\label{spectral_bounds}
It holds
\begin{equation*}
\widetilde{\lambda_0}\ge\lambda_0.
\end{equation*}
\end{lemma}
\begin{proof}
See Appendix \ref{app4}.
\end{proof}

Thanks to the above two Lemmas it becomes now clear why we work under Assumption \ref{HYP_POSITIVITY_PHI_K} in Theorem 
\ref{TEO_EXPLICIT_FORM_HUMAN_CAPITAL}. It provides a sufficient condition for the condition $r>\lambda_0$, imposed in Proposition \ref{step1}, to hold. In fact, assume $\widetilde K>0$ as in Assumption \ref{HYP_POSITIVITY_PHI_K}, then by Lemmas \ref{lem_k} and \ref{spectral_bounds} we immediately get $r>\lambda_0$.
\begin{remark}
Notice that, if $\Phi$ is a \emph{positive} measure, then $\widetilde K \equiv K$, $\lambda_0 \equiv \widetilde {\lambda_0}$ and the condition $K>0$ becomes also necessary, that is $K>0 \iff r>\lambda_0$.
\end{remark}

\section{Applications}
\label{Applications}\subsection{Optimal portfolio problems with path dependent labor income}

As a first practical application of our results, we consider the case in which the 
contingent claim $y$, given in \eqref{DYN_LABOR_INCOME_DELAY_I}, represents stochastic wages received by an agent
over her lifetime. As mentioned in the Introduction, stochastic delay differential equations allow for a realistic description of the  labor income evolution in continuous time models.
The path-dependency of the dynamics of the labor income \eqref{DYN_LABOR_INCOME_DELAY_I} is in line with the empirical literature showing that wages adjust slowly to financial markets shocks. In fact, as outlined by Keynes in the article \textit{The General Theory of Employment, Interest and Money} (1936), wages and prices do not adjust immediately to shocks in the economy. There is a vast literature on the topic and we refer to \cite{BGP} for a comprehensive list  of relevant papers. 

In this context, expression \eqref{HUMAN_CAPITAL} 
represents the present value of future discounted labor income and it has to be interpreted as the market value of the agent's human capital (see e.g.,  \cite{DL2010}). In Theorem \ref{TEO_EXPLICIT_FORM_HUMAN_CAPITAL} we therefore obtain a closed form solution for human capital, which makes
explicit the contributions of the market value of the past and the present. 
Our valuation formula results in
an explicit expression of human capital demonstrating the importance of appreciating
the past to quantify the current market value of future labor income. 

Whereas Assumption \ref{HYP_POSITIVITY_PHI_K} is all we need to provide the explicit valuation result of Theorem~\ref{TEO_EXPLICIT_FORM_HUMAN_CAPITAL}, the particular application to  
human capital 
requires labor income to be positive almost surely. A sufficient condition for this to be the case is 
provided in the following remark (for the proof see \cite[Proposition 2.7]{BGP}). 
\begin{remark} 
\label{rem_pos}
When $\varphi_i=0$ for all $i=1...n$, that is when the delay term in the volatility coefficient of \eqref{DYN_LABOR_INCOME_DELAY_I} vanishes, the variation of constants formula yields an explicit representation for $y$:
 \begin{eqnarray}	y(t)= \mathcal E(t)\big(   x_0 + \mathcal I(t)\big), 
	\end{eqnarray}
where 
\begin{eqnarray*}
		\mathcal E(t):= e^{(\mu_y -  \frac{1}{2}|\sigma_y|^2)t+ \sigma_y Z(t) },
		\qquad
		\mathcal I(t) := \int_0^t    \mathcal E^{-1}(u) \int_{-d}^0 \,y(s+u) \phi({\rm d}s)\, {\rm d}u.
	\end{eqnarray*} 
In this case one can see that, if $x_0>0$, $x_1 \ge 0$ a.s., $\phi \ge 0$ a.s., then $y(t)>0$ $\mathbb{P}$-a.s.	
\end{remark}
 
The results of Theorem \ref{TEO_EXPLICIT_FORM_HUMAN_CAPITAL} and the approach followed in this paper show how tools from infinite-dimensional analysis can be successfully used to address valuation problems that are non-Markovian, and hence beyond the reach of conventional approaches. 
For instance, Theorem \ref{TEO_EXPLICIT_FORM_HUMAN_CAPITAL} becomes an essential ingredient for solving some optimal control problems as the ones addressed in \cite{BGP}, \cite{DGZZ} and \cite{BiaGozZan}. 
In the following Examples we briefly recall the results of these paper, that show how the findings of Theorem \ref{TEO_EXPLICIT_FORM_HUMAN_CAPITAL} (or suitable generalization of it) can be successfully used to address an interesting class of optimal control problems. 

\begin{example}
\label{es1}
In \cite{BGP} the authors consider, and completely solve, an infinite horizon portfolio problem with borrowing constraints, in which an agent
receives labor income which adjusts to financial market shocks in a path dependent way. We briefly describe the framework of \cite{BGP}, emphasizing how our Theorem \ref{TEO_EXPLICIT_FORM_HUMAN_CAPITAL} is crucial in the derivation of their results.

In the framework of a Black\&Scholes complete market model described by \eqref{DYNAMIC_MARKET}, a representative agent  is endowed with initial wealth $w\geq 0$ and receives wages till her death.  The time of death $\tau_{\delta}$ is modeled as an exponential random variable of parameter $\delta >0$.
The wealth of the agent at time $t\ge 0$ is denoted by $W(t)$  and the wage rate is $y(t)$. She can invest in the riskless and risky assets, and can consume at a rate $ c(t)\geq 0$. The wealth allocated to the risky assets is $ \theta(t)\in \mathbb R^n$  at each time $t\geq 0$.  The agent has a bequest target $B(\tau_\delta)$ at  death, where the bequest  process $ B(\cdot)\geq 0$  is also chosen by the agent. To cover the gap between bequest and wealth at death: $ B(\tau_{\delta})-W(\tau_{\delta}),$ the agent pays an instantaneous life insurance premium of  $\delta( B(t)-W(t))$ for $t<\tau_\delta$ (for more details see \cite{DL2010}).
The agent's wealth (before death) is assumed to obey to the standard dynamic budget constraint of the Merton portfolio model, but with the labor income
	and insurance premium
	terms  added in the drift, exactly as in \cite{DL2010}; thus the agent's wealth $W$ satisfies the SDE
        \begin{equation}
          \label{eq:W_single}
          \begin{cases}
          dW(t) =  \left[W(t) r + \theta(t)\cdot (\mu-r\mathbf{1})  + y(t) - c(t)
-\delta\left(B(t)-W(t)\right)\right] {\rm d}t + \theta(t)\cdot\sigma {\rm d}Z(t),
\\
W(0)=w .
\end{cases}
        \end{equation}
In line with the empirical findings recalled above, to reflect a realistic economic setting, the dynamics of labor incomes adjust slowly to financial market shocks and it is modeled as a path-dependent delayed diffusion process of the form
\begin{equation}\label{delay}
\begin{cases}
{\rm d}y(t) =\left[\mu_y y(t)+\int_{-d}^0 \phi(s) y(t+s) {\rm d}s  \right] {\rm d}t + y(t)\sigma_y {\rm d}Z(t),
\\
y(0)=  x_0, \quad y(s) = x_1(s) \mbox{ for $s \in  [-d,0)$},
\end{cases}
\end{equation}
where $\mu_y \in  \mathbb R$, $\sigma_y \in  \mathbb R^n$, the functions $\phi(\cdot), x_1(\cdot)$ belong to $L^2\left(-d,0; \mathbb R\right)$ (and thus \eqref{delay} is a particular case of \eqref{DYN_LABOR_INCOME_DELAY_I}). 
\\
Denoted by $k>0$ the intensity of preference for leaving a bequest, $\gamma \in (0,1) \cup (1, +\infty)$ the risk-aversion coefficient and $\rho >0$ the discount rate, the aim is to maximize the expected utility from lifetime consumption and bequest,
\begin{eqnarray}\label{DEF_OBJECTIVE FUNCTION_DEATH TIME}
\mathbb E \left(\int_{0}^{+\infty} e^{-(\rho+ \delta) t }
\left( \frac{c(t)^{1-\gamma}}{1-\gamma}
+ \delta \frac{\big(k B(t)\big)^{1-\gamma}}{1-\gamma}\right) {\rm d}t
\right),
\end{eqnarray}
over all triplets $\left(c,\theta,B\right)\in \Big\{\mathbb F-\mbox{predictable} \ c(\cdot), B(\cdot), \theta(\cdot) \colon c(\cdot), B(\cdot) \in L^1 (\Omega \times [0, +\infty);\mathbb R_{+}),\theta(\cdot) \in L^2(\Omega \times \mathbb R; \mathbb R^n)\Big\}
$
satisfying the state constraint 
\begin{equation}\label{NO_BORROWING_WITHOUT_REPAYMENT_CONDITIONLA_MEAN}
W(t) +   \xi^{-1}(t)\mathbb E\left( \int_t^{+\infty} \xi(u) y(u) {\rm d}u \Bigg\vert \mathcal F_t\right)  \geq 0,
\end{equation}
which is a no-borrowing-without-repayment constraint
as the second term in
\eqref{NO_BORROWING_WITHOUT_REPAYMENT_CONDITIONLA_MEAN}
represents the agent's market value of
human capital at time $t$.  Human capital can be pledged as collateral,
and represents the agent's maximum borrowing capacity.
Notice that here $\xi$ satisfies equation \eqref{DYN_STATE_PRICE_DENSITY}, with a drift of the form $-\xi(t)(r+\delta)$, as explained in Remark \ref{remark:bounded_horizon}.
\\
Under such budget constraint, the authors find an explicit solution to expected power utility maximization from consumption and bequest. The proof of the result relies on the resolution of an infinite-dimensional Hamilton-Jacobi-Bellman (HJB) equation, that can be considered an infinite-dimensional version of the classical Merton problem. From a technical point of view, the key idea is to extend the state space so to include the past path of $y$. In this way the problem becomes infinite dimensional and Markovian in the current wealth and in the path of $y$ over the time window $[-d,0]$. In this infinite-dimensional reformulation of the problem it becomes essential to rewrite the constraint \eqref{NO_BORROWING_WITHOUT_REPAYMENT_CONDITIONLA_MEAN} by decomposing it in its present and past components. This kind of decomposition, provided by our Theorem \ref{TEO_EXPLICIT_FORM_HUMAN_CAPITAL}, is the essential ingredient that allows to find an explicit solution to the HJB equation, which in turns allows to completely solve the problem and to find explicitly the optimal
controls in feedback form (see \cite[Theorem 5.1, Section 5]{BGP}). \end{example}
 \begin{example}
 \label{es2}
 In \cite{BiaGozZan} is studied a robust version of the life-cycle optimal portfolio choice problem presented in Example \ref{es1}. Theorem \ref{TEO_EXPLICIT_FORM_HUMAN_CAPITAL} is needed in order to face the infinite-dimensional robust Merton problem the authors derive and to obtain an explicit solution of it. 
 \end{example}
 
 \begin{example}
 \label{es_3}
Another generalization of the problem addressed in \cite{BGP} is considered in \cite{DGZZ}. Similarly to \cite{BGP}, the authors consider a life-cycle optimal portfolio choice problem faced by an agent receiving labor income and allocating her wealth to risky assets and a riskless bond subject to a borrowing constraint. However here the dynamics of the labor income has two main features. First, labor income adjust slowly to financial market shocks, as in \cite{BGP}. Second, the labor income $y_i$ of an agent $i$ is benchmarked against  the labor incomes of a population $y^n:=(y_1,y_2,\ldots,y_n)$ of $n$ agents with comparable tasks and/or ranks. This last feature is faced taking the limit when $n\to +\infty$
so that the problem falls into the family of optimal control of infinite dimensional McKean-Vlasov Dynamics. The problem in studied in a simplified case where, adding a suitable new variable, the authors are able to find explicitly the solution of the associated infinite-dimensional HJB equation and find the optimal feedback controls. A necessary step to solve the problem is to provide a suitable reformulation of the no-borrowing without repayment constraint \eqref{NO_BORROWING_WITHOUT_REPAYMENT_CONDITIONLA_MEAN}, where now the labor income $y$ follows a stochastic delay differential equation where the drift contains not only a path-dependent term but also a mean reverting term. This issue is analyzed in \cite[Section 3]{DGZZ} where the authors provide a suitable generalization of formula \eqref{SOLUTION} carefully readapting the technique we use in Section \ref{proof_section}.   
\end{example}


\subsection{Counterparty risk and derivatives valuation}
\begin{example}
As a simple example of application of our setup to the context of over-the-counter derivatives, 
in equation \eqref{DYN_LABOR_INCOME_DELAY_I} consider the case of $n=1$,
$\mu_0=0$, $\phi=0$, $\sigma_0=0$, and
$ \varphi(s) =\delta_{-d}(s)$, where $\delta_{a}(s)$ indicates the delta-Dirac measure at $a$,
so that equation   \eqref{DYN_LABOR_INCOME_DELAY_I} reads
\begin{equation}\label{EXAMPLE_NEG}
\text{d}X_0(t) = X_0(t-d) \text{d} Z(t).
\end{equation}
Then, for $t \in [0,d )$
we have
\begin{eqnarray}
 X_0(t) = x_0+  \int_0^t X_0(s-d)         \text{d} Z(s)=
 x_0+  \int_{-d}^{t-d} x_1(\tau)          \text{d}Z(\tau+d) .
 \end{eqnarray}
In this case $X_0(t)$ is Gaussian, and dynamics \eqref{EXAMPLE_NEG} could be used to model, for example,  the variation margin of an over-the-counter swap,
when the collateralization procedure relies on a delayed mark-to-market
value of the instrument
(see \cite{brigo2013counterparty},  page~316, or \cite{BrigoPallavicini2014}, for example).
\end{example}

\section{Conclusion}
\label{sec:conclusions}

In recent years mathematical finance literature has seen the development of dynamics models that take into consideration the influence of past events on the current and future state of the system. Systems with delay find their well-deserved place in finance since delay in the dynamics can represent memory or inertia in the financial system.
On the other hand, this kind of models are not new and they appear in many applications, think for instance to population dynamics models in biology where delays occur naturally for biological reasons, or to epidemic models where time delays are introduced to model constant sojourn times in a state, for example, the infective or immune state. 
The applications studied in the literature generally make rather clear that, when introduced in an explicit way, time delays may change the qualitative behavior of a model.

From a mathematical perspective, delay systems has, in general, an infinite-dimensional nature. Problems connected to equations with delays are thus a challenging research area since the natural general approach to solve them is infinite-dimensional.

The findings of this paper perfectly fit into this research area. On the one hand, our result highlights, once again, how the knowledge of the past affects the knowledge of the future. On the other hand, it points out how tools from infinite-dimensional analysis can be effectively used to address problems involving delays. 

\subsection{Concluding remarks}

In this paper we have considered a standard Black\&Scholes complete market model where securities evolve as geometric Brownian motions. Despite the classical setup we started from, we have been able to derive an explicit pricing formula for stream of payments with a \textit{delayed} dynamics, by means of techniques from infinite-dimensional analysis. 
Our valuation formula results in an explicit expression demonstrating the importance of appreciating the past to quantify the current
market value of the future. The approach followed in this paper
highlights how tools from infinite-dimensional analysis can be successfully used to address valuation problems that are non-Markovian, and hence beyond the reach of conventional approaches. Moreover, as highlighted in Section \ref{Applications}, our results and the techniques we developed here, have already been successfully used to explicitly solve an interesting class of infinite dimensional stochastic optimal control problems.

\subsection{Future work}
From a financial perspective it may be interesting to address the problem of what happens if a path dependent object like \eqref{DYN_LABOR_INCOME_DELAY_I} is made available for trade. A problem of this kind has actually already been studied in \cite{arriojas2007delayed}. There, the authors assume that the stock prices satisfy stochastic functional differential equations and derive an explicit formula for the valuation of an European call option on a given stock. 
The problem addressed in \cite{arriojas2007delayed} could be consider a sort of counterpart of our setup. In fact, there the setup is non standard since the authors consider securities with path dependent dynamics. However, once the market has proved to be complete with a path-dependent market price of risk, the pricing of an European call becomes reasonably standard.

The valuable contribution of \cite{arriojas2007delayed} lies in showing that, if the securities of the market have path-dependent dynamics, then the market price of risk will have a path-dependent structure. As a consequence, pricing formulas can be obtained working under the  objective measure $\mathbb{P}$ and discounting by a path-dependent state price process $\xi$, suitable related to the market-price of risk.

Bearing in mind this finding, we could argue as in \cite{arriojas2007delayed} to see that, if an object of the form \eqref{DYN_LABOR_INCOME_DELAY_I} is made available for trade, then the market price of risk would change becoming path-dependent. The formal proof of this fact is outside the scope of this paper and it would be nothing innovative being just an extension of the results contained in \cite{arriojas2007delayed}. Anyway, the setup considered in \cite{arriojas2007delayed} raises up some interesting questions since one may now ask how the valuation formula we derived in Theorem \ref{TEO_EXPLICIT_FORM_HUMAN_CAPITAL} becomes if the stock prices have themselves a delayed dynamics. 
That is, suppose to consider, on the filtered probability space $(\Omega, \mathcal F, \mathbb F, \mathbb P)$, the $\mathbb F$-adapted vector valued process $(S_0,S)$, representing the price evolution of a money market account, $S_0$ with rate of return $r\ge 0$, and a single (for simplicity, as in \cite{arriojas2007delayed}) asset whose price $S(s)$ at time $s$ satisfies the following stochastic differential equation with delay:
\begin{equation}
\label{S_path_dep}
\begin{cases}
	{\rm }dS(s) = \left[ \int_{-d}^0  S(s+\tau) \,(\phi^S+\mu\delta_0)(\text{d}\tau) \right]\text{d}s
    +\left[ \int_{-d}^0S(s+\tau) (\phi_1^S+\sigma\delta_0)(\text{d}\tau)
  \right]  \text{d} Z(s),
  \\
S(0)= s_0, 
\\
S(s)= s_1(s )  \quad \mbox{  for $s \in [-d,0)$},
\end{cases}
   \end{equation}
where $\mu$ and $\sigma$ are positive constants, $\phi^S$ and $\phi_1^S$ are measure of bounded variation on $[-d,0]$, $\delta_0$ is the delta Dirac function and $Z$ is a one-dimensional Brownian motion (assume that $\mathbb F:=\{\mathcal F_t\}_{t\geq0}$ is the filtration
generated by $Z$, and enlarged with the
$\mathbb P$-null sets).
Arguing as in \cite{arriojas2007delayed} one should expect the market price of risk to be of the following form: \begin{footnote}
{We emphasize that the formal proof passes through an application of the Girsanov Theorem. In order for the hypothesis of the Girsanov Theorem to be satisfied, it is reasonable to impose suitable restrictions on $\mu$, $\sigma$, $\phi^S$ and $\phi_1^S$. For instance, a sufficient condition for \eqref{mpr} to be well defined is to assume $\phi^S_1\equiv 0$, $s_0>0$, $s_1 \ge 0$ a.s., $\phi^S \ge 0$ a.s. which ensures (see Remark \ref{rem_pos}) that $S(t)>0$ $\mathbb{P}$-a.s.. Here we just proceed heuristically.
}
\end{footnote}
\begin{equation}
\label{mpr}
\kappa(t):= -\frac{\int_{-d}^0 S(t+s) [\phi^S+(\mu-r)\delta_0]({\rm d}s)}{\int_{-d}^0S(t+s)\, (\phi_1^S+\sigma\delta_0)({\rm d} s)}.
\end{equation}

Thus the question one may ask is how the valuation formula we derived in Theorem \ref{TEO_EXPLICIT_FORM_HUMAN_CAPITAL} changes if the market price of risk has a stochastic and path-depentent form like in \eqref{mpr} (or like the one derived in \cite{arriojas2007delayed} if one wants to work in their setup). 
\begin{footnote}
{
Let us point out that a different issue is to understand if $y$ itself is tradable (for instance in the case it is not regarded as a stream of cashflows).  In this case it is not difficult to verify that, an object $y$ of the form \eqref{DYN_LABOR_INCOME_DELAY_I} would be tradable under a delayed Sharpe ratio condition of the following form:
\begin{equation}
\label{src}
\frac{\int_{-d}^0y(t+s)(\phi+(\mu_y-r)\delta_0)({\rm d} s)}{\int_{-d}^0y(t+s)(\varphi_1+\sigma_y\delta_0)({\rm d}s)}
=
\frac{\int_{-d}^0S(t+s)(\phi^S-(\mu+r)\delta_0)({\rm d} s)}{\int_{-d}^0S(t+s)(\phi_1^S+ \sigma\delta_0)({\rm d}s)}.
\end{equation}
Notice that, when the delay part is set to zero we recover the classical Sharpe ratio condition $\frac{\mu-r}{\sigma}=\frac{\mu_0-r}{\sigma_0}$. Moreover, mutatis mutandis, \eqref{src} has the same structure of the Sharpe ratio condition derived in \cite{arriojas2007delayed}.
}
\end{footnote}
This problem somehow would combine the novelties presented in \cite{arriojas2007delayed} and the ones presented in this paper since both the stock prices and the stream of payments would have a stochastic delayed dynamics. So the problem would be to derive an explicit pricing formula for \eqref{HUMAN_CAPITAL}, with $y$ as in \eqref{DYN_LABOR_INCOME_DELAY_I}, when the securities have themselves a path dependent dynamics of the form \eqref{S_path_dep} or the form considered in  \cite{arriojas2007delayed}. This is not a trivial extension of our result Theorem \ref{TEO_EXPLICIT_FORM_HUMAN_CAPITAL} because of the stochasticity and path dependency of the market price of risk. We leave it for future work.

\newpage

\appendix
\section{Proof of some results}

\subsection{Proof of Proposition \ref{well_pose}}
\label{app0}
The existence and uniqueness result for \eqref{DYN_LABOR_INCOME_DELAY_I} is not covered by the extant literature. When the initial datum $x = (x_0,x_1)$, seen as a function on $[-d,0]$, is  continuous, existence and uniqueness of the strong solution to the SDDE for $y$ is proved by \cite[Theorem I.2]{MOHAMMED_BOOK_96}. 
 When the initial datum  $x\in \R\times L^2([-d,0),{\rm d}t ;\R)$,  with the additional requirement that  $\phi$ and $(\varphi_1,...\varphi_n)$ are absolutely continuous w.r.t the Lebesgue measure, that is ${\rm d}\phi = \varphi \rm{d}t$, ${\rm d}(\varphi_1,...\varphi_n)= (\phi_1,...\phi_n)\,{\rm d}t$ and the Radon-Nikodym densities  $\varphi, \phi_1,...\phi_n \in L^2([-d,0), {\rm d}t)$, the existence and uniqueness result follows by \cite[Remark I.3(iv)]{MOHAMMED_BOOK_96}.  We need to extend this latter result to the case in which $\phi$ and  $(\varphi_1,...\varphi_n)$  are signed measures of bounded variation on $\left[-d,0\right]$. We will prove the result by means of the same procedure employed for the proof of \cite[Proposition B.2]{BiaGozZan}. There the authors prove the existence and uniqueness of the solution for an equation similar to \eqref{DYN_LABOR_INCOME_DELAY_I}, under more general assumptions on the measure $\phi$, but with no delay in the diffusion term.

Let us start by introducing the standard notation for the past path at $t$ of a (deterministic) function $h:[-d,T] \rightarrow \mathbb{R}$, for $0\leq t\leq T$, that is the function $h_t$
 \begin{equation*}
  h_t(s) := h(t+s) \, \,   \, \text{ for } -d \leq s \leq 0.
 \end{equation*}
    The past path of $y$ at $t$ for the realization $\omega$ is thus $y_t(s,\omega): =y(t+s,\omega) \, \, s \in [-d,0]$. As usual, we hide the dependence of the process on $\omega$ and write the delay terms in the drift and in the diffusion as follows:     \begin{equation}
    \label{D1}
    \int_{-d}^0y(t+s)\phi({\rm d}s) =\int_{-d}^0y_t(s)\phi({\rm d}s)
    \end{equation}
     and
  \begin{equation}
    \label{D2}
  \begin{pmatrix}
 \int_{-d}^0 y(t+s)  \varphi_1(\text{d}s) \\
 \vdots
 \\
 \int_{-d}^0y(t+s) \varphi_n(\text{d}s)
\end{pmatrix}    
=
\begin{pmatrix}
 \int_{-d}^0 y_t(s)  \varphi_1(\text{d}s) \\
 \vdots
 \\
 \int_{-d}^0y_t(s) \varphi_n(\text{d}s).
\end{pmatrix}    
    \end{equation}
       
The delay parts in \eqref{DYN_LABOR_INCOME_DELAY_I}, given by \eqref{D1} and \eqref{D2} can be then expressed in terms of (an extension of) the following linear operators ok kernel type:
\begin{equation}
\label{L}
L:C([-d,0];\mathbb{R}) \rightarrow \mathbb{R}, \qquad Lf:=\int_{-d}^0f(s)\, \phi({\rm d}s),
\end{equation}
and
\begin{equation}
\label{G2}
G:C([-d,0];\mathbb{R}) \rightarrow \mathbb{R}^n, \qquad Gf:= \begin{pmatrix}
 \int_{-d}^0 f(s)  \varphi_1(\text{d}s) \\
 \vdots
 \\
 \int_{-d}^0 f(s) \varphi_n(\text{d}s)
\end{pmatrix}.
\end{equation}
Since the operators $L$ and $G$ are well-defined only on the space of continuous functions $C([-d,0]; \mathbb{R})$, when the initial datum does not belongs to $C([-d,0]; \mathbb{R})$ but just to $L^2([-d,0); \mathbb{R})$ problems may arise.  
In fact, consider the initial datum $(x_0, x_1) \in \mathbb{R} \times L^2(-d,0;\mathbb{R})$ and proceed by assuming that the solution to \eqref{DYN_LABOR_INCOME_DELAY_I} exists.  Denote the past path on the window $[t-d,t]$ by $y_t:[-d,0] \rightarrow \mathbb{R}$, $y_t(s):=y(t+s)$ a.e. $t \ge 0$, $s \in [-d,0]$. Then, for $0 \le t <d$, the past path is
\begin{equation*}
y_t(s)=
\begin{cases}
y(t+s) &\text{if} \ -t \le s<0
\\
x_1(s) &\text{if} -d \le s<-t.
\end{cases}
\end{equation*} 
which, in general, is not a continuous function, but only square integrable. Therefore, the operators $L$ and $G$ introduced in \eqref{L} and \eqref{G2} cannot be applied to $y_t$ since the integrals in \eqref{D1} and \eqref{D2} may not be well defined.
\newline
The first issue is thus to show that $L$ and $G$ admit continuous extensions to the square integrable functions on $[-d,0]$, as made precise in the following lemma.
\begin{lemma}
\label{pre_lem}
Let $L:C([-d,0];\mathbb{R}) \rightarrow \mathbb{R}$ and $G: C([-d,0];\mathbb{R}) \rightarrow \mathbb{R}^n$ be the continuous and linear maps given in \eqref{L} and \eqref{G2} respectively. Fix $T>0$ and define on $C([-d,T];\mathbb{R})$ the delay operators
\begin{equation*}
\mathcal{L}(y)(t):=Ly_t, \qquad 0 \le t\le T,
\end{equation*}
\begin{equation*}
\mathcal{G}(y)(t):=Gy_t, \qquad 0 \le t\le T.
\end{equation*}
Then,
\begin{enumerate}
\item the maps $\mathcal{L}:C([-d,T];\mathbb{R}) \rightarrow L^2([0,T];\mathbb{R})$ and $\mathcal{G}:C([-d,T];\mathbb{R}) \rightarrow L^2([0,T];\mathbb{R}^n)$ satisfy, respectively, the inequalities
\begin{equation}
\label{starL}
\|\mathcal{L}(y)\|_{L^2([0,T];\mathbb{R})}\le|\phi|([-d,0]) \|y\|_{L^2([-d,T];\mathbb{R})}, \qquad \forall y \in C([-d,T];\mathbb{R}).
\end{equation}
\begin{equation}
\label{starG}\|\mathcal{G}(y)\|_{L^2([0,T];\mathbb{R}^n)}\le\left( \sum_{i=1}^n [|\varphi_i|([-d,0])]^2\right)^{\frac 12} \|y\|_{L^2([-d,T];\mathbb{R})}, \qquad \forall y \in C([-d,T];\mathbb{R}).
\end{equation}
\item $\mathcal{L}$ and $\mathcal{G}$ have $L^2$-norm continuous, linear extensions (still denoted by $\mathcal{G}$ and $\mathcal{L}$, respectively) to $L^2([-d,T];\mathbb{R})$. 
\end{enumerate}
\end{lemma}
\begin{proof}
The proof follows the lines of \cite[Lemma B.1]{BiaGozZan} (see also \cite[Part II, Chapter 4, Theorem 3.3]{BENSOUSSAN_DAPRATO_DELFOUR_MITTER}, but for the sake of completeness, we prove the result for the operator $\mathcal{G}$. For the operator $\mathcal{L}$ one follows the same reasoning.
\begin{enumerate}
\item
\begin{align*}
\|\mathcal{G}(y)\|_{L^2([0,T];\mathbb{R}^n)}
&=\|Gy_{\cdot}\|_{L^2([0,T];\mathbb{R}^n)}
= \sup_{h\in L^2([0,T];\mathbb{R}^n), \|h\|_{L^2}\le 1}\int_0^TGy_r \cdot h(r)\, {\rm d}r 
\\
&=\sup_{h\in L^2([0,T];\mathbb{R}^n), \|h\|_{L^2}\le 1}\int_0^T
\begin{pmatrix}
 \int_{-d}^0 y_r(s)  \varphi_1(\text{d}s) \\
 \vdots
 \\
 \int_{-d}^0y_r(s) \varphi_n(\text{d}s)
\end{pmatrix}    \cdot h(r)\, {\rm d}r 
\\
&= \sup_{h\in L^2([0,T];\mathbb{R}^n), \|h\|_{L^2}\le 1}\sum_{i=1}^{n} \int_0^T h_i(r)\int_{-d}^0 y_r(s)\varphi_i({\rm d}s)\, {\rm d}r.
\end{align*}
We estimate the i-th component ($i=1...n$) of the above expression exploiting the Fubini Theorem and the H\"older inequality.
\begin{align*}
\int_0^Th_i(r) \int_{-d}^0y_r(s)\, \varphi_i({\rm d}s)\, {\rm d}r
& \le \int_0^T|h_i(r)| \int_{-d}^0|y(r+s)|\, |\varphi_i|({\rm d}s)\, {\rm d}
r
\\
&= \int_{-d}^0\int_0^T|h_i(r)||y(r+s)|\,  {\rm d}r\, |\varphi_i|({\rm d}s)
\\
&\le\int_{-d}^0\|h_i\|_{L^2([0,T];\mathbb{R})}\|y\|_{L^2([s,s+T];\mathbb{R})}\, |\varphi_i|({\rm d}s)
\\
&\le |\varphi_i|([-d,0])\|h_i\|_{L^2([0,T];\mathbb{R})}\|y\|_{L^2([-d,T];\mathbb{R})},
\end{align*}
where for the last inequality we exploit the inclusion $[s, s+T] \subseteq [-d,T]$.
Therefore, by means of the H\"older inequality we obtain
\begin{align*}
\|\mathcal{G}(y)\|_{L^2([0,T];\mathbb{R}^n)}
&\le \sup_{h\in L^2([0,T];\mathbb{R}^n), \|h\|_{L^2}\le 1} \sum_{i=1}^{n}|\varphi_i|([-d,0])\|h_i\|_{L^2([0,T];\mathbb{R})}\|y\|_{L^2([-d,T];\mathbb{R})}
\\
&\le \|y\|_{L^2([-d,T];\mathbb{R})}\sup_{h\in L^2([0,T];\mathbb{R}^n), \|h\|_{L^2}\le 1} \left( \sum_{i=1}^n [|\varphi_i|([-d,0])]^2\right)^{\frac 12}\left( \sum_{i=1}^n \|h_i\|^2_{L^2([0,T];\mathbb{R})}\right)^{\frac 12}
\\
&=\|y\|_{L^2([-d,T];\mathbb{R})}\sup_{h\in L^2([0,T];\mathbb{R}^n), \|h\|_{L^2}\le 1} \left( \sum_{i=1}^n [|\varphi_i|([-d,0])]^2\right)^{\frac 12}\|h\|^2_{L^2([0,T];\mathbb{R}^n)}
\\
&\le \left( \sum_{i=1}^n [|\varphi_i|([-d,0])]^2\right)^{\frac 12}\|y\|_{L^2([-d,T];\mathbb{R})}.
\end{align*}
\item The existence of the bounded linear extension of $\mathcal{L}$ and $\mathcal{G}$ to $L^2([-d,t];\mathbb{R})$ is a consequence of inequalities \eqref{starL} and \eqref{starG} and the fact that $C([-d,T];\mathbb{R})$ is dense in $L^2([-d,T];\mathbb{R})$.
\end{enumerate}
\end{proof}
We are now ready to prove Proposition \ref{well_pose}.
\begin{proof}[Proof of Proposition \ref{well_pose}]
The proof of the result relies on a contraction type argument. The same argument has been used in the proof of \cite[Proposition B.2]{BiaGozZan}. There the authors consider a SDDE of type \eqref{DYN_LABOR_INCOME_DELAY_I} with no delay in the diffusion term. On the other hand they work in a more general setting considering a measure valued process $\phi$ in the delay integral of the drift term. 
\newline
We provide here a sketch of the proof referring to \cite{BiaGozZan} for more details. We will give just the details of the estimates concerning the delay integral in the diffusion term that is missing in \cite{BiaGozZan}. 

Let us fix the initial condition $(x_0, x_1) \in \mathbb{R}^2 \times L^2([-d,0];\mathbb{R})$. Let $T>0$, we introduce the space
\begin{equation*}
S_T:=\{y \in C([0,T];\mathbb{R}):y(0)=x_0\},
\end{equation*}
endowed with the sup norm 
\begin{equation*}
\|y\|_{S_T} = \sup_{t \in [0,T]}|y(t)|.
\end{equation*}
We consider the space $L^p(\Omega;S_T)$, $p \ge 2$, endowed with the norm 
\begin{equation*}
\|y\|_{L^p(\Omega;S_T)} = \left(\mathbb{E}\left[\|y\|^p_{S_T}\right]\right)^{\frac 1p}= \left( \mathbb{E}\left[\sup_{t \in [0,T]}|y(t)|^p\right]\right)^{\frac 1p}.
\end{equation*}
In the sequel we will denote by $p':=\frac{p}{p-1}$ the conjugate exponent to $p$ and by $p^*:=\frac{p}{p-2}$ the conjugate exponent to $\frac p2$.
\newline
Given $y \in L^p(\Omega;S_T)$, let
\begin{align}
\label{F}
F(y)(t):= x_0+ \mu_y \int_0^t y(r)\,{\rm d}r + \int_0^t \mathcal{L}(\bar y^{x_1})\, {\rm d}r+ \int_0^t y(r)\sigma_y \cdot {\rm d}Z(r)+ \int_0^t\mathcal{G}(\bar y^{x_1})\cdot {\rm d}Z(r), \qquad 0 \le t \le T.
\end{align}
Here $\mathcal{L}$ and $\mathcal{G}$ are the continuous linear operators introduced in Lemma \ref{pre_lem} and $\bar y^{x_1} \in L^p(\Omega;L^2([-d,T];\mathbb{R}) )$ is defined as follows:
\begin{equation}
\label{glu}
 \bar{y}^{x_1}(t) = 
 \begin{cases}
 x_1(t), &  \text{if }    -d\leq t< 0; \\
y(t),  &  \text{if }  \  \  0\leq t\leq T.
\end{cases}
\end{equation}
We aim at proving that $F$ maps $L^p(\Omega;S_T)$ into itself for any $p\ge 2$ and that it is a contraction on the same space when $p>4$.

Let us start by proving that $F$ maps $L^p(\Omega, S_T)$, $p\ge 2$, into itself.
We write 
\begin{align}
\label{con1}
\|F(y)\|_{L^p(\Omega;S_T)}
& \le |x_0|+  |\mu_y|\left\Vert\int_0^{\cdot}y(r)\, {\rm d}r\right\Vert_{L^p(\Omega;S_T)}
\notag\\
& + \left\Vert\int_0^{\cdot} \mathcal{L}(\bar y^{x_1})\, {\rm d}r \right\Vert_{L^p(\Omega;S_T)}
 +\left\Vert\int_0^{\cdot} y(r)\sigma_y \cdot \, {\rm d}Z(r)\right\Vert_{L^p(\Omega;S_T)}
 + \left\Vert \int_0^{\cdot} \mathcal{G}(\bar y^{x_1}) \cdot \, {\rm d}Z(r)\right\Vert_{L^p(\Omega;S_T)}.
\end{align}
The boundedness of the terms that appears in the r.h.s. of \eqref{con1}, except the last one, can be proved following the lines of \cite[Proposition B.2]{BiaGozZan}. We estimate the last term in the r.h.s. of \eqref{con1} by means of the Burkholder-Davies-Gundy inequality
\begin{align*}
\left\Vert \int_0^{\cdot} \mathcal{G}(\bar y^{x_1}) \cdot \, {\rm d}Z(r)\right\Vert_{L^p(\Omega;S_T)}^p
&= \mathbb{E} \left[\sup_{t \in [0,T]}\left|\int_0^t \mathcal{G}(\bar y^{x_1})\cdot {\rm d}Z(r)\right|^p \right]
\lesssim  \mathbb{E} \left[\left|\int_0^T\| \mathcal{G}(\bar y^{x_1})\|^2\, {\rm d}r\right|^{\frac p2} \right]
\\
&= \mathbb{E} \left[\| \mathcal{G}(\bar y^{x_1})\|^p_{L^2([0,T];\mathbb{R}^n)}\right]
\lesssim \mathbb{E} \left[\|\bar y^{x_1}\|^p_{L^2([-d,T];\mathbb{R})}\right]
\\
&=\|\bar y^{x_1}\|^p_{L^p(\Omega;L^2([-d,T];\mathbb{R}))} < \infty,
\end{align*}
where in the last inequality we exploited \eqref{starG} of Lemma \ref{pre_lem}.

Let us now prove that, for $p>4$, $F$ defines a contraction in $L^p(\Omega, S_T)$. We endow this space by the equivalent norm

\begin{equation}
\label{a_norm}
\|y\|_{\alpha} := \left(\mathbb{E} \left[\sup_{t \in [0,T]}\left(e^{-\alpha t}|y(t)|\right)^p\right]\right)^{\frac1p},
\end{equation}
where $\alpha>0$ will be chosen later on. Once we proved that $F$ defines a contraction, by the Banach fixed point Theorem, we can infer the existence of a unique $y \in L^p(\Omega;S_T)$ such that $y=F(y)$, i.e. 
\begin{align*}
y(t)= x_0+ \mu_y \int_0^t y(r)\,{\rm d}r + \int_0^t \mathcal{L}(\bar y^{x_1})\, {\rm d}r+ \int_0^t y(r)\sigma_y \cdot {\rm d}Z(r)+ \int_0^t\mathcal{G}(\bar y^{x_1})\cdot {\rm d}Z(r), \qquad 0 \le t \le T, \qquad \mathbb{P}-a.s.,
\end{align*}
and this will conclude the proof.

Given $y, z \in L^p(\Omega;S_T)$, from \eqref{F} and \eqref{a_norm}, we have
\begin{align}
\label{dif}
 \|F(z)-F(y)\|_{\alpha}^p
 \lesssim_p
&  
\mathbb{E}
  \left[\sup_{t \in [0,T]}
  e^{-p\alpha t}\left( |\mu_y|\left|\int_0^{t}(z(r)-y(r)){\rm d}r\right|^p
  +
 \left|\int_0^{t} \mathcal{L}(\bar z^{x_1}-\bar y^{x_1})\, {\rm d}r \right|^p \right)\right]
\notag
\\
&+\mathbb{E}\left[\sup_{t \in [0,T]} e^{-p\alpha t}\left|\int_0^{t} (z(r)-y(r))\sigma_y \cdot\, {\rm d}Z(r)\right|^p\right]
\notag
\\
&+\mathbb{E}\left[\sup_{t \in [0,T]} e^{-p\alpha t}\left|\int_0^{t} \mathcal{G}(\bar z^{x_1}-\bar y^{x_1})\cdot\, {\rm d}Z(r)\right|^p\right].
\end{align}
We can estimate the first three terms in the r.h.s. of \eqref{dif} proceeding as in \cite[Proposition B.2]{BiaGozZan} \begin{footnote}{For more details on the estimates, the interested reader can consult that paper.}\end{footnote}.
For the first term we obtain
\begin{align}
\label{T_1}
 \mathbb{E}\left[\sup_{t \in [0,T]}
  e^{-p\alpha t} |\mu_y|\left|\int_0^{t}(z(r)-y(r)){\rm d}r\right|^p\right]
  \le |\mu_y|T\left(\frac{1}{\alpha p'}\right)^{\frac{p}{p'}}\|z-y\|^p_{\alpha}
  \lesssim_{\mu_y,T, p} C_1(\alpha)\|z-y\|^p_{\alpha}.
\end{align}
For the second term we get
\begin{equation}
\label{T_2}
\mathbb{E}
  \left[\sup_{t \in [0,T]}
  e^{-p\alpha t}
 \left|\int_0^{t} \mathcal{L}(\bar z^{x_1}-\bar y^{x_1})\, {\rm d}r \right|^p \right]
\le \left(\frac{|\phi|([-d,0])}{\alpha p'} \right)^{\frac{p}{p'}}T |\phi|([-d,0])\|z-y\|^p_{\alpha}
\lesssim_{|\phi|, p, T}C_2(\alpha)\|z-y\|^p_{\alpha}.
 \end{equation}
 For the third term, by means of the so called factorization method (see e.g. \cite[Section 5.3]{DAPRATO_ZABCZYK_BOOK_1}), for a given $\delta \in \left(\frac 1p, \frac 12\right)$\begin{footnote}{Notice that this condition require to work with $p>2$.}\end{footnote}, we have
\begin{align}
\label{T_3}
\mathbb{E}&\left[\sup_{t \in [0,T]} e^{-p\alpha t}\left|\int_0^{t} (z(r)-y(r))\sigma_y \cdot\, {\rm d}Z(r)\right|^p\right]
\notag\\
&\lesssim_{p, \delta} 
\left(\int_0^T u^{p'(\delta-1)}e^{-p'\alpha u}\, {\rm d}u \right)^{\frac{p}{p'}}T \|\sigma_y\|^p\left( \sup_{u \in [0,T] }\int_0^u(u-r)^{-2\delta}e^{-2\alpha(u-r)}\, {\rm d}r\right)^{\frac p2}\|z-y\|_{\alpha}^p
\notag \\
&\lesssim_{p,\delta, T, \|\sigma_y\|} C_3(\alpha)\|z-y\|^p_{\alpha}.
\end{align}

Let us now come to the estimate of the fourth term in \eqref{dif}. Exploiting the factorization method, for $\eta \in \left( \frac 1p, \frac{p-2}{2p}\right)$ \begin{footnote}{This condition is made in order to guarantee the convergence of the integrals that will appear in what follows. Notice that this condition require to work with $p>4$.}\end{footnote} we rewrite that stochastic integral as follows
\begin{equation*}
\int_0^t \mathcal{G}(\bar z^{x_1}-\bar y^{x_1}) \cdot \,{ \rm d}Z(r)=c_{\eta}\int_0^t (t-u)^{\eta-1} Y(u)\, {\rm d}u, 
\end{equation*}
with
\begin{equation*}
\frac{1}{c_{\eta}}:= \int_r^t (t-u)^{\eta-1}(u-r)^{-\eta}\, {\rm d}u = \frac{\pi}{\sin(\pi\eta)},
\end{equation*}
and
\begin{equation*}
Y(u)=\int_0^u (u-r)^{-\eta}\mathcal{G}(\bar z^{x_1}-\bar y^{x_1}) \cdot {\rm d}Z(r).
\end{equation*}

Thanks to the H\"older inequality we estimate 
\begin{align*}
e^{-\alpha t}\left|\int_0^t \mathcal{G}(\bar z^{x_1}-\bar y^{x_1})\cdot {\rm d}Z(r)\right|
&= c_{\eta} e^{-\alpha t}\left| \int_0^t (t-u)^{\eta-1}Y(u)\, {\rm d}u \right|
\\
&= c_{\eta} \left| \int_0^t e^{-\alpha (t-u)}(t-u)^{\eta-1}e^{-\alpha u}Y(u)\, {\rm d}u \right|
\\
&\le c_{\eta}\left(\int_0^te^{-\alpha p'(t-u)}(t-u)^{p'(\eta-1)}\, {\rm d}u\right)^{\frac{1}{p'}} \left(\int_0^te^{-\alpha p u}|Y(u)|^p\, {\rm d}u\right)^{\frac 1p}.
\end{align*}
Therefore we obtain
\begin{align*}
\mathbb{E}&\left[\sup_{t \in [0,T]}e^{-\alpha p t}\left| \int_0^t \mathcal{G}(\bar z^{x_1}-\bar y^{x_1}) \cdot {\rm d}Z(r) \right|^p \right]
\\
&\le c_{\eta}^p \mathbb{E}\left[\sup_{t \in [0,T]}
\left(\int_0^te^{-\alpha p'(t-u)}(t-u)^{p'(\eta-1)}\, {\rm d}u\right)^{\frac{p}{p'}} \left(\int_0^te^{-\alpha p u}|Y(u)|^p\, {\rm d}u\right)
 \right]
\\
&\le c^p_{\eta} \left(\int_0^Te^{-\alpha p'u}u^{p'(\eta-1)}\, {\rm d}u\right)^{\frac{p}{p'}} \int_0^T e^{-\alpha pu} \mathbb{E} \left[ |Y(u)|^p\right]\, {\rm d}u.
 \end{align*}
Now, recalling the definition of $\mathcal{G}$ and that, when $r<d$, $\bar z^{x_1}_r(s)-\bar y^{x_1}_r(s)=0$ for $s \in [-d,-r)$ (see \eqref{glu}), by means of the Burkholder-Davis-Gundy (BDG) and the H\"older (H) inequalities, we obtain for all $u \in [0,T]$,
\begin{align*}
e^{-\alpha pu}&\mathbb{E}\left[ |Y(u)|^p\right] 
= e^{-\alpha pu} \mathbb{E}\left[ \left| \int_0^u (u-r)^{-\eta}\mathcal{G}(\bar z^{x_1}-\bar y^{x_1}) \cdot {\rm d}Z(r)\right|^p \right]
\\
&\overset{BDG}{\lesssim_p} e^{-\alpha up} \mathbb{E}\left[ \left| \int_0^u (u-r)^{-2\eta}\|\mathcal{G}(\bar z^{x_1}-\bar y^{x_1})\|^2\, {\rm d}r\right|^{\frac p2}\right]
\\
&=e^{-\alpha up}\mathbb{E}\left[ \left| \int_0^u (u-r)^{-2\eta}
\sum_{i=1}^n \left|\int_{-d}^0(\bar z_r^{x_1}-\bar y_r^{x_1})(s)\, \varphi_i({\rm d}s) \right|^2
\, {\rm d}r\right|^{\frac p2}\right]
\\
&=e^{-\alpha up}\mathbb{E}\left[ \left| \int_0^u (u-r)^{-2\eta}
\sum_{i=1}^n \left|\int_{-d\vee -r}^0\left((z(r+s)-y(r+s)\right)\, \varphi_i({\rm d}s) \right|^2
\, {\rm d}r\right|^{\frac p2}\right]
\\
&=\mathbb{E}\left[ \left| \int_0^u (u-r)^{-2\eta}e^{-2\alpha (u-r-s)}e^{-2\alpha (r+s)}\sum_{i=1}^n \left|\int_{-d\vee -r}^0\left(z(r+s)-y(r+s)\right)\, \varphi_i({\rm d}s) \right|^2
\, {\rm d}r\right|^{\frac p2}\right]
\\
&\overset{H}\le \mathbb{E}\left[ \left| \int_0^u (u-r)^{-2\eta}e^{-2\alpha (u-r-s)}e^{-2\alpha (r+s)}\sum_{i=1}^n |\varphi_i|([-d,0]) \int_{-d\vee -r}^0\left|(z(r+s)-y(r+s)\right|^2\, \varphi_i({\rm d}s)
\, {\rm d}r\right|^{\frac p2}\right]
\\
&\overset{H}{\le}\left( \sum_{i=1}^n \left(|\varphi_i|([-d,0]) \right)^{p^*}\int_0^u\int_{-d \vee -r}^0(u-r)^{-2p^*\eta}e^{-2\alpha p^*(u-r-s)}\, \varphi_i({\rm d}s)\, {\rm d}r\right)^{\frac{p}{2p^*}} 
\\
&\qquad \quad \ \mathbb{E} \left[\sum_{i=1}^n\int_0^u \int_{-d \vee -r}^0 e^{-\alpha p(r+s)}|z(r+s)-y(r+s)|^p\, \varphi_i({\rm d}s)\, {\rm d}r \right]
\\
&\le \left( \sum_{i=1}^n \left(|\varphi_i|([-d,0]) \right)^{p^*}\int_0^u\int_{-d \vee -r}^0(u-r)^{-2p^*\eta}e^{-2\alpha p^*(u-r-s)}\, \varphi_i({\rm d}s)\, {\rm d}r\right)^{\frac{p}{2p^*}} 
\\
&\qquad \quad \ \sum_{i=1}^n\int_0^u \int_{-d \vee -r}^0\mathbb{E} \left[ \sup_{(r+s) \in [0,u]}\left(e^{-\alpha p(r+s)}|z(r+s)-y(r+s)|^p\right)\right]\, \varphi_i({\rm d}s)\, {\rm d}r
\\
&\le \left( \sum_{i=1}^n \left(|\varphi_i|([-d,0]) \right)^{p^*}\int_0^u\int_{-d \vee -r}^0(u-r)^{-2p^*\eta}e^{-2\alpha p^*(u-r-s)}\, \varphi_i({\rm d}s)\, {\rm d}r\right)^{\frac{p}{2p^*}} 
\\
& \qquad \quad \ u  \sum_{i=1}^n |\varphi_i|([-d,0]) \mathbb{E} \left[ \sup_{(r+s) \in [0,u]}\left(e^{-\alpha p(r+s)}|z(r+s)-y(r+s)|^p\right)\right].
\end{align*}
Therefore, 
\begin{align*}
\int_0^T &e^{-\alpha pu}\mathbb{E}\left[ |Y(u)|^p\right] \, {\rm d}u
\\
&\lesssim_{|\varphi_i|,p} 
\int_0^Tu \left(\sum_{i=1}^n \int_0^u\int_{-d \vee -r}^0(u-r)^{-2p^*\eta}e^{-2\alpha p^*(u-r-s)}\, \varphi_i({\rm d}s)\, {\rm d}r\right)^{\frac{p}{2p^*}} 
\\
&\qquad \quad \mathbb{E} \left[ \sup_{(r+s) \in [0,u]}\left(e^{-\alpha p(r+s)}|z(r+s)-y(r+s)|^p\right)\right] \, {\rm d}u
  \\
  &\lesssim_{|\varphi_i|,p} 
\int_0^Tu \left( \sum_{i=1}^n \int_0^u\int_{-d \vee -r}^0(u-r)^{-2p^*\eta}e^{-2\alpha p^*(u-r-s)}\, \varphi_i({\rm d}s)\, {\rm d}r\right)^{\frac{p}{2p^*}}
\\
& \qquad \quad \ 
  \mathbb{E} \left[ \sup_{(r+s) \in [0,T]}\left(e^{-\alpha p(r+s)}|z(r+s)-y(r+s)|^p\right)\right] \, {\rm d}u
  \\
  &=\left(\int_0^Tu \left(\sum_{i=1}^n \int_0^u\int_{-d \vee -r}^0(u-r)^{-2p^*\eta}e^{-2\alpha p^*(u-r-s)}\, \varphi_i({\rm d}s)\, {\rm d}r\right)^{\frac{p}{2p^*}} \, {\rm d}u \right ) \ \|z-y\|_{\alpha}^p
  \\
 & \lesssim_{|\varphi_i|,T,p} \left(\int_0^Tr^{-2p^*\eta}e^{-2\alpha p^*r} \, {\rm d}r\right)^{\frac{p}{2p^*}}  \|z-y\|_{\alpha}^p
 \end{align*}
where the last inequality is obtained as follows:
\begin{align*}
\int_0^T&u  \left(\sum_{i=1}^n \int_0^u\int_{-d \vee -r}^0(u-r)^{-2p^*\eta}e^{-2\alpha p^*(u-r-s)}\, \varphi_i({\rm d}s)\, {\rm d}r\right)^{\frac{p}{2p^*}} \, {\rm d}u
\\
&\le T \int_0^T \left(\sum_{i=1}^n \int_0^u\int_{-d \vee -r}^0(u-r)^{-2p^*\eta}e^{-2\alpha p^*(u-r)} e^{2\alpha p^*s}\, \varphi_i({\rm d}s)\, {\rm d}r\right)^{\frac{p}{2p^*}} \, {\rm d}u
\\
&\le T \int_0^T \left(\sum_{i=1}^n \int_0^u\int_{-d \vee -r}^0(u-r)^{-2p^*\eta}e^{-2\alpha p^*(u-r)} \, \varphi_i({\rm d}s)\, {\rm d}r\right)^{\frac{p}{2p^*}} \, {\rm d}u
\\
&\le T \sup_{u \in[0,T]} \left(\sum_{i=1}^n |\varphi_i|([-d,0])\int_0^u(u-r)^{-2p^*\eta}e^{-2\alpha p^*(u-r)} \, {\rm d}r\right)^{\frac{p}{2p^*}} 
\\
&\lesssim_{|\varphi_i|, T}\left(\int_0^Tr^{-2p^*\eta}e^{-2\alpha p^*r} \, {\rm d}r\right)^{\frac{p}{2p^*}}.
\end{align*}
Putting together the above estimates we obtain
\begin{align}
\label{T_4}
\mathbb{E}&\left[\sup_{t \in [0,T]}e^{-\alpha p t}\left| \int_0^t \mathcal{G}(\bar z^{x_1}-\bar y^{x_1}) \cdot {\rm d}Z(r) \right|^p \right]
\notag\\
&\lesssim_{T, |\varphi_i|,\eta,p} \left(\int_0^Te^{-\alpha p'u}u^{p'(\eta-1)}\, {\rm d}u\right)^{\frac{p}{p'}} \left(\int_0^Tr^{-2p^*\eta}e^{-2\alpha p^*r} \, {\rm d}r\right)^{\frac{p}{2p^*}}  \|z-y\|_{\alpha}^p
\notag\\
&\lesssim_{T, |\varphi_i|, \eta,p} C_4(\alpha) \|z-y\|_{\alpha}^p,
\end{align}
Finally, from \eqref{T_1}, \eqref{T_2}, \eqref{T_3} and \eqref{T_4} we infer
\begin{equation*}
\|F(z)-F(y)\|^p_{\alpha} \lesssim_{\mu_y,T,p,|\phi|,|\varphi_i|, \|\sigma_y\|, \delta, \eta} \sum_{i=1}^4 C_i(\alpha) \|z-y\|_{\alpha}^p,
\end{equation*}
where $C_i(\alpha) \rightarrow 0$ as $\alpha\rightarrow \infty$, for  $i=1,...,4$.
So, by taking $\alpha>0$ sufficiently large, this proves that $F$ is a contraction, thus there exists a unique
fixed point of it. In this way we prove the existence and uniqueness of the solution in the space $L^p(\Omega,S_T)$ for $p >4$. Since, for such $p$,
$L^p(\Omega,S_T)\subset L^2(\Omega,S_T)$,
such solution also belongs to $L^2(\Omega,S_T)$.
To get uniqueness in the space $L^2(\Omega,C([0,T];\mathbb{R}))$
one can take two solutions $y$ and $\tilde y$ in this space and take their difference. Using the fact that both are fixed points of $F$, by means of the Gronwall Lemma one gets
$\sup_{t\in[0,T]} \mathbb{E}\left[|y(t)-\tilde y(t)|^2\right]=0$ and this concludes the proof.\end{proof}


\subsection{Proof of Lemma \ref{LEMMA_STOCHASTIC_INTEGRAL}}
 \label{app1}
 \begin{proof}
  Let us denote with $\sigma_y^i$ the $i$-th component of $\sigma_y$,
  and let us show that
    \begin{eqnarray*}
   \mathbb E \left[   \int_{t_0}^t  \left | y(s)\sigma_y^i  +
 \int_{-d}^0  y(s+\tau) \varphi_i(\text{d} \tau)\right|^2 \text{d}s  \right] < +\infty.
  \end{eqnarray*}
  By the trivial inequality $(a+b)^2 \leq 2( a^2 + b^2)$, it is sufficient to show that
   \begin{eqnarray}\label{EX_I}
   \mathbb E \left[  \int_{t_0}^t   |y(s)\sigma_y^i|^2  \text{d}s
 \right] < +\infty,
  \end{eqnarray}
  and
   \begin{eqnarray}\label{EX_II}
   \mathbb E \left[   
    \int_{t_0}^t
\left|\int_{-d}^0 y(s+\tau)  \varphi_i(\text{d}\tau) \right|^2 \text{d}s  \right] < +\infty.
  \end{eqnarray}
 We immediately see that \eqref{EX_I} holds true thanks to  Proposition \ref{well_pose}.
   
To show \eqref{EX_II}, we use the H\"older inequality and the Fubini Theorem to estimate
\begin{align*}
\mathbb{E} \left[\int_{t_0}^t\left| \int_{-d}^0 y(\tau+s)\, \varphi_i({\rm d}\tau)\right|^2\,{\rm d}s \right]
&\le |\varphi_i| ([-d,0]) \mathbb{E} \left[\int_{t_0}^t \int_{-d}^0 |y(\tau+s)|^2\, |\varphi_i|({\rm d}\tau)\,{\rm d}s \right]
\\
&=|\varphi_i| ([-d,0]) \mathbb{E} \left[\int_{-d}^0 \int_{t_0}^t |y(\tau+s)|^2\, {\rm d}s\, |\varphi_i|({\rm d}\tau)\right]
\\
&=|\varphi_i| ([-d,0]) \mathbb{E} \left[\int_{-d}^0 \int_{t_0+\tau}^{t+\tau} |y(r)|^2\, {\rm d}r\, |\varphi_i|({\rm d}\tau)\right]
\\
&\le |\varphi_i| ([-d,0]) \mathbb{E} \left[\int_{-d}^0 \int_{t_0-d}^{t} |y(r)|^2\, {\rm d}r\, |\varphi_i|({\rm d}\tau)\right]
\\
&=\left(|\varphi_i| ([-d,0])\right)^2 \mathbb{E} \left[ \int_{t_0-d}^{t} |y(r)|^2\, {\rm d}r)\right],
\end{align*}
which is finite, thanks to Proposition \ref{well_pose}.
\end{proof}

\subsection{Proof of Lemma \ref{LEMMA_RESOLVENT_SET_RESOLVENT_BAR_A}}
\label{app2}
\begin{proof}
If $\lambda \in \mathbb{R} \cap \rho(A)$ then $K(\lambda) \ne 0$ by Lemma \ref{lemma_speck}.
To compute $R(\lambda, A)$,
we will consider for a fixed $\mathbf{m}=\left(m_0,m_1\right)\in \mathcal{H}$
the equation
\begin{equation}\label{EQ_PROOF_RESOLVENT}
(\lambda-A)\left(u_0,u_1\right)=\left(m_0,m_1\right) ,
\end{equation}
in the unknown $(u_0,u_1)\in \mathcal{D}(A)$, that by definition of $A$ is equivalent to
\[
\left\{
\begin{aligned}
(\lambda- (\mu_y - \sigma_y \cdot \kappa) ) u_0- \int_{-d}^0u_1(\tau)\Phi(\text{d}\tau) &=m_0\\
\lambda u_1-\frac{\text{d}u_1}{\text{d}s}&=m_1.
\end{aligned}\right. \]
Then
\[u_1(s)=e^{\lambda s} u_0+\int_s^0e^{-\lambda(\tau-s)}m_1(\tau)\, \text{d}\tau,\quad s\in[-d,0],\]
and $u_0$ is determined by the equation
\[\big(\lambda- (\mu_y - \sigma_y \cdot \kappa)\big)u_0=\left[ m_0+\int_{-d}^0\left( e^{\lambda \tau} u_0
+\int_{\tau}^0e^{-\lambda(s-\tau)}m_1(s)\, \text{d}s \right)\Phi(\text{d}\tau)\right],\]
yielding
\begin{eqnarray*}
K(\lambda)u_0
=m_0+ \int_{-d}^0\int_{-d}^{s}e^{-\lambda (s-\tau)}\Phi(\text{d}\tau) \, m_1(s) \text{d} s.
\end{eqnarray*}
Then the result immediately follows.
\end{proof}

\subsection{Proof of Lemma \ref{lem_k}}
\label{app3}

\begin{proof}
It is immediate to check that the function $\widetilde K:\mathbb R\to\mathbb R$ is continuous and differentiable with
\[\widetilde K^\prime(\xi)=1+\int_{-d}^0e^{\xi\tau}|\tau|\,|\Phi|(\text{d}\tau)>0, 
\]
and that
\[\lim_{\xi\to\pm\infty}\widetilde K(\xi)=\pm\infty\,.\]
Equation $\widetilde K(\xi)=0$ has thus exactly one real solution $\bar {\xi}$. Let us now show that $\bar{\xi}=\widetilde{\lambda_0}$. By the definition of $\widetilde{\lambda_0}$, clearly
we have that $\bar{\xi}\le\widetilde{\lambda_0}$. To show the opposite inequality, $\bar{\xi}\ge \widetilde{\lambda_0}$, we consider an arbitrary $\lambda=a+i b\in \mathbb{C}$ such that $\widetilde K(\lambda)=0$. 
Then
\[\begin{aligned}
0&=\text{Re}(\widetilde{K}(\lambda))=a-\mu_y+\sigma_y\cdot \kappa-\int_{-d}^0e^{a\tau}\cos(b\tau)\,|\Phi|(\text{d}\tau)\\
&\ge a-\mu_y+\sigma_y \cdot \kappa-\int_{-d}^0e^{a\tau}\,|\Phi|(\text{d}\tau)
=:\widetilde K(a)\,.
\end{aligned}\]
Since $\widetilde{K}$ is an increasing function, we can infer Re$\lambda\le \bar{\xi}$ and taking the supremum in the definition of $\widetilde{\lambda_0}$ we obtain $\widetilde{\lambda_0} \le \bar{\xi}$.
By the same argument, the relation $\widetilde K(r) > 0 \iff r > \widetilde{\lambda_0}$ immediately follows.
\end{proof}

\subsection{Proof of Lemma \ref{spectral_bounds}}
\label{app4}
\begin{proof}
Since by Lemma \ref{lem_k} we know that $\widetilde{K}$ is an increasing function and $\widetilde{K}(\widetilde{\lambda_0})=0$, we just need to prove that $\widetilde{K}(\lambda_0) \le 0$.
For every $\lambda=a+i b \in \mathbb{C}$ we have 
\begin{equation*}
\text{Re}(K(\lambda))=a-(\mu_y+\sigma_y\cdot\kappa)-\int_{-d}^0 e^{a\tau}\cos(b\tau)\,\Phi( {\rm d}\tau).
\end{equation*}
Recalling the definition of $\lambda_0$ it is enough to show that, for every $\lambda=a+i b \in \mathbb{C}$ such that $K(\lambda)=0$, it holds $\widetilde{K}(a)\le 0$.
We have that 
\begin{align*}
\tilde K (a)
&=a-(\mu_y+\sigma_y\cdot \kappa)-\int_{-d}^0 e^{a\tau}|\Phi|({\rm d}\tau)
\\
&=a-(\mu_y+\sigma_y\cdot \kappa)-\int_{-d}^0 e^{a\tau}\cos(b\tau)\,\Phi({\rm d}\tau) -\int_{-d}^0 e^{a\tau}|\Phi|({\rm d}\tau) + \int_{-d}^0 e^{a\tau}\cos(b\tau)\,\Phi({\rm d}\tau)
\\
&\le\text{Re}(K(\lambda))+\int_{-d}^0e^{a\tau} \, [\Phi-|\Phi|]({\rm d}\tau) \le 0.
\end{align*}
 This concludes the proof.
\end{proof}


\end{document}